\documentclass[pdflatex,sn-mathphys-num]{sn-jnl}% Math and Physical Sciences Numbered Reference Style
\usepackage{graphicx}%
\usepackage{multirow}%
\usepackage{amsmath,amssymb,amsfonts}%
\usepackage{amsthm}%
\usepackage{mathrsfs}%
\usepackage[title]{appendix}%
\usepackage{xcolor}%
\usepackage{textcomp}%
\usepackage{manyfoot}%
\usepackage{booktabs}%
\usepackage{algorithm}%
\usepackage{algorithmicx}%
\usepackage{algpseudocode}%
\usepackage{listings}%

\usepackage{braket}
\usepackage{bm}
\usepackage{tikz}
\usetikzlibrary{arrows.meta}

\newtheorem{lemma}{Lemma}
\newtheorem{theorem}{Theorem}
\newtheorem{corollary}{Corollary}

\newtheorem{example}{Example}
\theoremstyle{thmstyleone}%
\theoremstyle{thmstyletwo}%

\theoremstyle{thmstylethree}%

\begin{document}

\title{Kirkwood-Dirac classical states based on discrete Fourier transform: Representation with directed graph}

%%=============================================================%%
%% GivenName	-> \fnm{Joergen W.}
%% Particle	-> \spfx{van der} -> surname prefix
%% FamilyName	-> \sur{Ploeg}
%% Suffix	-> \sfx{IV}
%% \author*[1,2]{\fnm{Joergen W.} \spfx{van der} \sur{Ploeg} 
%%  \sfx{IV}}\email{iauthor@gmail.com}
%%=============================================================%%

\author[1]{\fnm{Lin-Yan} \sur{Cai}}\email{cailinyan2025@163.com}

\author*[2]{\fnm{Ying-Hui} \sur{Yang}}\email{yangyinghui4149@163.com}
\equalcont{These authors contributed equally to this work.}

\author[1]{\fnm{Zhu-Jun} \sur{Zheng}}\email{zhengzj@scut.edu.cn}
\equalcont{These authors contributed equally to this work.}

\affil[1]{\orgdiv{School of Mathematics}, \orgname{South China University of Technology}, \city{Guangzhou}, \postcode{510640}, \country{China}}

\affil*[2]{\orgdiv{School of Mathematics and Information Science}, \orgname{Henan Polytechnic University}, \orgaddress{\city{Jiaozuo}, \postcode{454000}, \state{State}, \country{China}}}

%%==================================%%
%% Sample for unstructured abstract %%
%%==================================%%

\abstract{
The Kirkwood-Dirac (KD) quasiprobability distribution is a fundamental representation for quantum states and has been widely applied in quantum metrology, quantum chaos, weak values in recent years. A quantum state is KD-classical if its KD-quasiprobability distribution forms a valid classical probability distribution with respect to two given bases, and KD-nonclassical otherwise, with the latter being closely associated with quantum advantages in various quantum  processes. In this work, we investigate the structural characteristics of the KD-classical state set when the transition matrix between two orthonormal bases takes the form of a discrete Fourier transform (DFT) matrix. First, we adopt an alternative analytical approach to prove that the set of KD-classical states in a $p^r$-dimensional Hilbert space is the convex hull of KD-classical pure states--a conclusion that was recently established by De Bi{\`e}vre et al [Annales Henri Poincar{\'e}, 1-20, 2025]. Furthermore, we define a directed graph and use it to characterize KD-classical pure states in a Hilbert space of arbitrary dimension $d$. That is, the convex hull of KD-classical pure states along any path from the start vertex to the end vertex in this directed graph is exactly the intersection of the KD-classical state set and the linear space spanned by these path-associated KD-classical pure states. This general result not only yields the $p^r$-dimensional conclusion in a straightforward manner but also encompasses Theorem 2 in the existing work [J. Phys. A, 57, 435303, 2024], demonstrating its generality and inclusiveness.}

%\keywords{Kirkwood-Dirac classical state, discrete Fourier transform, KD quasiprobability distribution, convex combination, directed graph}

%%\pacs[JEL Classification]{D8, H51}

%%\pacs[MSC Classification]{35A01, 65L10, 65L12, 65L20, 65L70}

\maketitle

\section{Introduction}
Quantum mechanics is distinguished from classical mechanics by a host of inherent non-classical properties, including quantum coherence \cite{baumgratz2014quantifying}, quantum entanglement \cite{horodecki2009quantum}, and quantum discord \cite{ollivier2001quantum}. In general, a physical process is deemed classical if it can be fully described by a joint probability distribution. In contrast, processes that can only be characterized by quasiprobability distributions are regarded as unique quantum mechanical phenomena \cite{schmid2024kirkwood,spekkens2008negativity}. Formally, quasiprobability distributions satisfy the normalization condition analogous to classical probability distributions, yet they violate certain Kolmogorov axioms and may take negative or non-real values. As a powerful analytical tool for investigating the non-classical nature of quantum mechanics, quasiprobability distributions have been widely studied, with the Winger function \cite{wigner1932quantum} being the most famous representative. The Winger function maps quantum states onto the classical phase space (position-momentum space) to depict the state of a quantum system. \par

With the rapid development of modern quantum mechanics, especially in the field of quantum computing, discrete quantum systems (e.g., qubits) have found increasingly extensive application, and the Kirkwood-Dirac (KD) quasiprobability distribution has emerged as a more suitable tool for describing such system. The KD-quasiprobability distribution was first proposed by Kirkwood in 1935 \cite{Kirkwood1935statistical}, and subsequently, Dirac independently introduced it in 1945 \cite{Dirac1945analogy} to express quantum non-classicality. Compared to the Wigner quasiprobability distribution, the KD-quasiprobability distribution allows the use of arbitrary pairs of observables to describe a quantum state. \par

Driven by the progress of quantum mechanics research and the practical demands of quantum computing, KD-quasiprobability distribution has become a prominent research focus in quantum information. The authors in Ref.\cite{arvidsson2020quantum,zheng2026kirkwood} have revealed the intrinsic correlation between KD-nonclassicality and quantum advantages in various
quantum tasks. To date, KD-distribution has been extensively applied to the exploration of quantum coherence \cite{budiyono2023quantifying,budiyono2024quantum}, quantum contextuality \cite{schmid2024kirkwood,spekkens2008negativity}, and quantum entanglement \cite{budiyono2025quantum}, and it also serves as a core tool in studies of weak values \cite{dressel2015weak,pusey2014anomalous,lupu2022negative}, quantum precision measurement \cite{arvidsson2020quantum,jenne2022unbounded}, quantum tomography \cite{lundeen2012procedure,thekkadath2016direct}, quantum chaos \cite{yunger2018quasiprobability,gonzalez2019out}, quantum thermodynamics \cite{lostaglio2018quantum,levy2020quasiprobability,lostaglio2020certifying,lostaglio2023kirkwood,gherardini2024quasiprobabilities}, fundamental principles of quantum mechanics \cite{hofmann2012complex,halliwell2016leggett,rastegin2023kirkwood}, quantum steering \cite{zheng2026kirkwood}. Experimentally, researchers in Ref.\cite{bargmann1964note,fernandes2024unitary,mukunda2001bargmann} have successfully measured the KD-quasiprobability distribution of quantum states.
Moreover, the expression of the KD-distribution in terms of Bargmann invariants \cite{bargmann1964note,fernandes2024unitary,mukunda2001bargmann} allows for numerical measurements via quantum circuits \cite{wagner2024quantum,oszmaniec2024measuring,quek2024multivariate}. These results lay a solid experimental foundation for investigating quantum system properties through the KD-quasiprobability distribution. For a comprehensive overview of the diverse research directions related to the KD-quasiprobability distribution, readers are referred to the review article \cite{arvidsson2024properties}.\par

In 2021, Arvidsson-Shukur et al. \cite{arvidsson2021conditions} established a sufficient condition for the nonclassicality of the KD-quasiprobability distribution. That is, the KD-quasiprobability distribution of a pure quantum state $\ket{\psi}$ contains negative or non-real values if $n_{a}(\psi)+n_{b}(\psi)>\frac{3}{2}d$, where $n_{a}(\psi)$ and $n_{b}(\psi)$ denote the number of non-zero coefficients when $\ket{\psi}$ is expanded in orthonormal bases $\big\{\ket{a_i} \big\}$ and $\big\{\ket{b_i}\big\}$, respectively. Subsequently, De Bi{\`e}vre \cite{de2021complete,de2023relating} investigated scenarios involving completely incompatible observables and uncovered the connection between support
uncertainty and KD-nonclassicality, proving that the KD-quasiprobability distribution of $\ket{\psi}$ also exhibits negative or non-real values when $n_{a}(\psi)+n_{b}(\psi)>d+1$ and the transition matrix between the two given bases has nonzero entries. Langrenez et al. \cite{langrenez2024characterizing} were the first to characterize the structure of mixed states within the set of KD-classical states under specific conditions. They demonstrated that the set of all KD-classical states in a prime-dimensional Hilbert space is exactly the convex hull of two  families of basis states when the transition matrix between the two bases is the discrete Fourier transform (DFT) matrix. In 2024, Xu \cite{xu2024kirkwood} derived a necessary and sufficient criterion for KD-classical pure states, which provides a complete characterization of the structural properties of all KD-classical pure states. Yang et al. \cite{yang2024geometry} proposed a conjecture that the set of KD-classical states is precisely the convex hull of KD-classical pure states when the transition matrix is the DFT matrix, and they provided a proof for this conjecture in Hilbert spaces of dimension $p^2$ (where $p$ is a prime). Xu \cite{xu2024hermitian} obtained relevant results for real operators associated with the DFT matrix, and De Bi{\`e}vre \cite{de2025kirkwood} proved that the validity of the  set of KD-classical states is precisely the convex hull of KD-classical pure states when the basis transition matrix is the DFT matrix in prime power-dimensional Hilbert spaces by virtue of group representation and there exists a counter-example for $d = 6$. Additionally, Spriet \cite{spriet2025characterizing} investigated the KD-classicality on second countable locally compact abelian groups. \par

In this work, we focus on the KD-classical states for which the transition matrix is the DFT matrix. We first adopt an alternative approach to prove that the set of KD-classical states is equal to the convex hull of KD-classical pure states in $p^r$-dimensional Hilbert space (Theorem \ref{thm1}). For arbitrary $d$-dimensional Hilbert spaces, we then define a directed graph and establish a novel result (Theorem \ref{dthm}): the convex hull of the KD-classical pure states on any path from the start vertex to the end vertex in this directed graph is the intersection of the set of KD-classical states and the linear space spanned by these KD-classical pure states on the path. We show
that both Theorem \ref{thm1} and Theorem 2 in Ref.\cite{yang2024geometry} are special cases of our Theorem \ref{dthm}.\par

The remainder of this paper is organized as follows. In Section \ref{Preliminaries}, we introduce the fundamental definitions and notations used throughout this paper. We characterize the geometric properties of KD-classical states in $p^r$-dimensional Hilbert space in Section \ref{prime-power}. In Section \ref{any-d}, we give  directed graph construction and establish Theorem \ref{dthm}  for arbitrary $d$-dimensional Hilbert space. Conclusions and discussions are given in Section \ref{conclusion-diss}.

\section{Preliminaries}\label{Preliminaries}

Consider a $d$-dimension Hilbert space $ \mathcal{H}$ with orthonormal bases $\mathcal{A}=\big\{\ket{a_i}\big\}$ and
$\mathcal{B}=\big\{\ket{b_j}\big\}$. In terms of these two bases $\mathcal{A}$ and $\mathcal{B}$, the KD distribution of a quantum state $\rho$ can be written as
\begin{equation*}
    Q_{ij}(\rho)=\bra{a_j}\rho\ket{b_i}\braket{b_i|a_j}.
\end{equation*}

The KD quasiprobability distribution satisfies the following properties:
\begin{equation*}
    \sum_{i,j}Q_{ij}(\rho)=1.
\end{equation*}
Similar to the normalization property of classical probability distributions—with the key distinction that the KD quasiprobability distribution takes values in the complex number field $\mathbb{C}$.\par
A quantum state $\rho$ with respect to bases $\mathcal{A}$ and $\mathcal{B}$ is called  KD-classical if all values of its KD quasiprobability distribution are non-negative. It means that the KD distribution of KD-classical state $\rho$ constitutes a classical probability distribution.\par
For the two orthonormal bases $\mathcal{A}$ and $\mathcal{B}$, the transition matrix $U$ between them is termed a discrete Fourier transform (DFT) matrix when its elements satisfy the condition 
\begin{equation*}
    U_{ij}=\braket{a_i|b_j}=\frac{1}{\sqrt{d}}\omega_d^{ij}.
\end{equation*}
where $\omega_d=e^{\frac{2\pi\sqrt{-1}}{d}}$, $d$ is the dimension of Hilbert space.\par
As established in the previous work \cite{xu2024kirkwood,yang2024geometry}, when the transition matrix $U$ is the DFT matrix, the corresponding KD-classical pure states are given by

\begin{equation}
    \ket{\psi_{ms}}=\frac{1}{\sqrt{x}}\sum_{k=0}^{x-1}\omega_x^{sk}\ket{a_{ky+m}}=\frac{1}{\sqrt{y}}\omega_d^{-ms}\sum_{l=0}^{y-1}\omega_y^{-ml}\ket{b_{lx+s}}, m\in \mathbb{Z}_y,s \in \mathbb{Z}_{x} \label{purekd}
\end{equation}
where $xy=d$. Note that for $x=d, y=1$ ,  the corresponding KD-classical pure state is $\{\ket{\psi_{ms}}\}=\{\ket{b_j}\}$, and $m=0,s=j$. Similarly for $x=1, y=d$, $\{\ket{\psi_{ms}}\}=\{\ket{a_i}\}$ with $m=i,s=0$.\par
With this notation the set $\mathrm{pure}(\mathrm{KD}_{\mathcal{A,B}}^+)$  of KD-classical pure state is

\begin{equation*}
    %\mathrm{pure}(KD_{\mathcal{A,B}}^+)=
    \mathrm{pure}(\mathrm{KD}_{\mathcal{A,B}}^+)=\bigg\{\ket{\psi_{ms}}\bra{\psi_{ms}}\quad\bigg|\quad \forall m\in \mathbb{Z}_y,s \in \mathbb{Z}_{x},xy=d\bigg\}.
\end{equation*}
Note that $\mathrm{KD}_{\mathcal{A,B}}^+$ denotes the set of KD-classical states. 
KD distribution of these KD-classical pure states has the following property.
\begin{lemma}\label{property}
    Suppose $d=xy$. If $i=m' \mod y$ and $j=s' \mod x$, then
    \begin{equation}\label{propofKD}
    \begin{aligned}
    \bra{a_i}\psi_{ms}\ket{b_j}\braket{b_j|a_i}=\left \{
    \begin{array}{ll}
    \frac{1}{d}\quad &if\quad m=m'\; and \;s=s'\\
    0&otherwise\\
   \end{array}
   \right.    
  \end{aligned}.
\end{equation}
\end{lemma}
The proof of Lemma \ref{property} is given in  \ref{prop}.\par

For the DFT matrix as the transition matrix, Yang et al\cite{yang2024geometry} proved $\mathrm{KD}_{\mathcal{A,B}}^+=\mathrm{conv}\left(\mathrm{pure}(\mathrm{KD}_{\mathcal{A,B}}^+)\right)$ when $d=p^2$, and conjectured this result holds for any $d$. Applying the finite group representation theory, De Bièvre et al \cite{de2025kirkwood} proved that this result holds only for $d=p^r$ but there exists a counter-example for $d=6$.
    
In the following, we will use an alternative method to prove $\mathrm{KD}_{\mathcal{A,B}}^+=\mathrm{conv}\left(\mathrm{pure}(\mathrm{KD}_{\mathcal{A,B}}^+)\right)$ in $d=p^r$ dimensional Hilbert space. Then for any $d$ we give a novel result about KD-classical states.

\def\QEDclosed{\mbox{\rule[0pt]{1.3ex}{1.3ex}}}
\def\QED{\QEDclosed}
\def\proof{\indent{\em Proof}.}
\def\endproof{\hspace*{\fill}~\QED\par\endtrivlist\unskip}

\section{KD-classical States in $p^r$ dimensional Hilbert space}\label{prime-power}%%

In this subsection, we will characterize the KD classicality of quantum states for DFT matrix as the transition matrix in a $p^r$-dimensional Hilbert space.

For a prime $p$, every element $a \in \mathbb{Z}_{p^r}$ can be uniquely expressed in its $p$-adic expansion as
\begin{equation}
    a=a_0p^{0}+a_1p+...+a_{r-1}p^{r-1},
\end{equation}
where the  $0\leq a_i\leq p-1$ for all $i$. This tuple $(a_0,a_1,...a_{r-1})$ uniquely determines the integer $a$. For simplicity, we denote integer $a$ by vector $\bm{a}_{0...r-1}$, i.e.,
\begin{equation}
    a=a_0p^{0}+a_{1}p+...+a_{r-1}p^{r-1}:=(a_0,a_{1}...a_{r-1}):=\bm{a}_{0...r-1}.\label{a-exp}
\end{equation}
Note that Eq.(\ref{a-exp}) can also be written as
\begin{equation*}
\begin{split}
    a=&(a_0p^{0}+...+a_{w-1}p^{w-1})+p^{w}(a_{w}p^0+...+a_{z}p^{z-w})+p^{z+1}(a_{z+1}p^{0}+...+a_{r-1}p^{r-z-2})\\
     =&(a_0,...a_{w-1})+p^{w}(a_{w},...a_{z})+p^{z+1}(a_{z+1},...a_{r-1})\\
     =&\bm{a}_{0...w-1}+p^{w}\bm{a}_{w...z}+p^{z+1}\bm{a}_{z+1...r-1}.
\end{split}
\end{equation*}
The difference between two integers $\bm{a}_{0...r-1}$ and $\bm{a}'_{0...r-1}$ can be denoted by $\Delta^{\bm{a}\bm{a}'}_{0...r-1}$, i.e.,
\begin{equation}\label{diffeqn}
    \Delta^{\bm{a}\bm{a}'}_{0...r-1}=\bm{a}_{0...r-1}-\bm{a}'_{0...r-1}=\Delta^{\bm{a}\bm{a}'}_{0...w-1}+p^{w}\Delta^{\bm{a}\bm{a}'}_{w...z-w}+p^{z+1}\Delta^{\bm{a}\bm{a}'}_{z+1...r-z-2}
\end{equation}
%For a vector $(a_0,a_1...a_{r-1})=\bm{a}_{0...r-1}$, we introduce the notation $\bm{a}_{w...z}$ with $0\leq w\leq z\leq r-1$ to denote the new vector consisting of the $w$-th through $z$-th components, and the integer of this new vector represent is :
%\begin{equation*}
%    \bm{a}_{w...z}=(a_w...a_z)=a_w+...+a_{z}p^{z-w}.
%\end{equation*}
%In the following discussion, the boldface vector corresponds to an integer.\par

We now return to the KD-classical pure states in Eq.(\ref{purekd}). For the dimension $d=p^r$, consider the indices $i=ky+m$ for the state $\ket{a_{ky+m}}$ and $j=lx+s$ for the state $\ket{b_{lx+s}}$. Suppose that $x=p^{\alpha},y=p^{\beta}$, then $\alpha+\beta=r$. Analyzing the $p$-adic representations of $i$ and $j$, we have
\begin{equation}
\begin{split}
    &i=m+kp^{\beta}=\bm{i}_{0...\beta-1}+p^{\beta}\bm{i}_{\beta...r-1}=\bm{i}_{0...r-1},\\
    &j=s+lp^{\alpha}=\bm{j}_{0...\alpha-1}+p^{\alpha}\bm{j}_{\alpha...r-1}=\bm{j}_{0...r-1},
\end{split} \label{ijms}
\end{equation}
where
\begin{equation*}
    \begin{split}
        m=\bm{i}_{0...\beta-1}\;,\;k=\bm{i}_{\beta...r-1}\;,\;
        s=\bm{j}_{0...\alpha-1}\;,\; l=\bm{j}_{\alpha...r-1}.
    \end{split}
\end{equation*}

Therefore, the KD-classical pure states in Eq.(\ref{purekd}) can be rewritten as
\begin{equation}\label{p_use}
\begin{split}
    \ket{\psi_{\substack{\bm{i}_{0...\beta-1}\\\bm{j}_{0...\alpha-1}}}}
    &=\frac{1}{\sqrt{p^{\alpha}}}\sum_{\bm{i}_{\beta...r-1}}^{}
    \omega_{p^{\alpha}}^{\bm{j}_{0...\alpha-1}\bm{i}_{\beta...r-1}}
    \ket{a_{\bm{i}_{0...\beta-1}+p^{\beta}\bm{i}_{\beta...r-1}}}\\
    &=\frac{1}{\sqrt{p^{\beta}}}
    \omega_d^{-\bm{i}_{0...\beta-1}\bm{j}_{0...\alpha-1}}
    \sum_{\bm{j}_{\alpha...r-1}}
    \omega_{p^{\beta}}^{-\bm{i}_{0...\beta-1}\bm{j}_{\alpha...r-1}}\ket{b_{\bm{j}_{0...\alpha-1}+p^{\alpha}\bm{j}_{\alpha...r-1}}}.
\end{split}
\end{equation}
Here we adopt the research methodology from Ref.\cite{langrenez2024characterizing,yang2024geometry} to investigate the relation among $\mathrm{span}_{\mathbb{R}}\{\mathrm{pure}(\mathrm{KD}_{\mathcal{A,B}}^+)\}$,  $\mathrm{KD}_{\mathcal{A,B}}^+$ and $\mathrm{conv}(\mathrm{pure}(\mathrm{KD}_{\mathcal{A,B}}^+))$.
Let us first consider a property on KD-classical pure states. For simplicity, we denote the projective operator $|\psi\rangle\langle\psi|$ by $\psi$ throughout the paper.

\begin{lemma}\label{lemma1}
The relation between KD-classical states corresponding to different factorizations of $p^r$ is
\begin{equation*}
    \begin{split}
        &\sum\limits_{j_{\alpha-1}=0}^{p-1}
\psi_{\substack{\bm{i}_{0...\beta-1}\\\bm{j}_{0...\alpha-1}}}
=\sum\limits_{i_{\beta}=0}^{p-1}
\psi_{\substack{\bm{i}_{0...\beta}\\\bm{j}_{0...\alpha-2}}},\quad 1\leq\alpha\leq r-1.
    \end{split}
\end{equation*}
\end{lemma}\par

Note that the non-bold letter $j_{\alpha-1}$ beneath the summation symbol denotes the $(\alpha-1)$-th component of vector $\bm{j}_{0...r-1}$.
The complete proof is provided in \ref{The proof of Lemma1}.

When $x=d, y=1$ ,  the corresponding KD-classical pure state is $\{\ket{\psi_{ms}}\}=\{\ket{b_j}=\psi_{0j}\}$. Similarly for $x=1, y=d$, $\{\ket{\psi_{ms}}\}=\{\ket{a_i}=\ket{\psi_{i0}}\}$ with $m=i,s=0$. Applying Lemma \ref{lemma1}, we can obtain the following corollary.
\begin{corollary}\label{corollary}
\begin{equation*}
    \begin{split}
        &\sum_{i_{r-1}}a_{\bm{i}_{0...r-1}}=\sum_{j_0}\psi_{\substack{\bm{i}_{0...r-2}\\j_0}}\\
        &\sum_{i_0}\psi_{\substack{i_0\\\bm{j}_{0...r-1}}}=\sum_{j_{r-1}}b_{\bm{j}_{0...r-1}}
    \end{split}
\end{equation*}
\end{corollary}

In the follows, we will use Lemma \ref{lemma1} to prove Lemma \ref{lemma2}.

\begin{lemma}\label{lemma2}
    $\mathrm{KD}_{\mathcal{A,B}}^{+}\cap \mathrm{span}_{\mathbb{R}}\left(\mathrm{pure}(\mathrm{KD}_{\mathcal{A,B}}^{+})\right)=\mathrm{conv}\left(\mathrm{pure}(\mathrm{KD}_{\mathcal{A,B}}^{+})\right)$ for dimension $d=p^r$
\end{lemma}

The complete proof of Lemma \ref{lemma2} is provided in \ref{The proof of Lemma2}.

The definition of the KD quasi-probability distribution can be naturally generalized to self-adjoint operators. For a self-adjoint operator $F$, we define
\begin{equation}\label{selfkd}
    Q_{ij}(F)=\bra{a_{i}}F\ket{b_{j}}\braket{b_{j}|a_{i}}.
\end{equation}
If $Q_{ij}(F)\in \mathbb{R}$ for all $i,j\in \mathbb{Z}_d$, we refer to the self-adjoint operator $F$ as a KD-real operator. The set of all such KD-real operators is denoted by $\mathrm{KD}_{\mathcal{A,B}}^r$. It is straightforward to verify that $\mathrm{KD}_{\mathcal{A,B}}^r$ forms a linear space.
The following Lemma establishes the connection between $\mathrm{KD}_{\mathcal{A,B}}^r$ and $\mathrm{KD}_{\mathcal{A,B}}^+$.

%\begin{lemma}\label{lemma3}
%   $\mathrm{KD}_{\mathcal{A,B}}^r=\mathrm{span}_{\mathbb{R}}\left(\mathrm{pure}(\mathrm{KD}_{\mathcal{A,B}}^+)\right)$
%   leads to
%   $\mathrm{KD}_{\mathcal{A,B}}^+=\mathrm{conv}\left(\mathrm{pure}(\mathrm{KD}_{\mathcal{A,B}}^+)\right)$
%\end{lemma}

%\begin{proof}
%    It is immediate that $\mathrm{conv}\left(\mathrm{pure}(\mathrm{KD}_{\mathcal{A,B}}^+)\right)\subseteq \mathrm{KD}_{\mathcal{A,B}}^+$. It therefore suffices to establish the converse inclusion. By lemma\ref{lemma2}, we have
%    \begin{equation*}
%        \forall \rho\in \mathrm{KD}_{\mathcal{A,B}}^+\subseteq \mathrm{KD}_{\mathcal{A,B}}
%    ^r=\mathrm{span}_{\mathbb{R}}\left(\mathrm{pure}(\mathrm{KD}_{\mathcal{A,B}}^+)\right),
%    \end{equation*}
 %   which means
%    \begin{equation*}
 %       \rho\in \mathrm{KD}_{\mathcal{A,B}}^{+}\cap \mathrm{span}_{\mathbb{R}}\left(\mathrm{pure}(\mathrm{KD}_{\mathcal{A,B}}^{+})\right)=\mathrm{conv}\left(\mathrm{pure}(\mathrm{KD}_{\mathcal{A,B}}^{+})\right).
%    \end{equation*}
% That means 
 %   \begin{equation*}
 %        \mathrm{KD}_{\mathcal{A,B}}^+\subseteq \mathrm{conv}\left(\mathrm{pure}(\mathrm{KD}_{\mathcal{A,B}}^+)\right).
 %   \end{equation*} 
%\end{proof}\par

Recently, Xu \cite{xu2024hermitian} and De Bièvre et al \cite{de2025kirkwood} independently obtained the following result.

\begin{lemma}\label{lemmakdr}
    $\mathrm{KD}_{\mathcal{A,B}}^r=\mathrm{span}_{\mathbb{R}}\left(\mathrm{pure}(\mathrm{KD}_{\mathcal{A,B}}^+)\right)$ for any dimension $d$.
\end{lemma}\par
Since $\mathrm{KD}_{\mathcal{A,B}}^+ \subseteq \mathrm{KD}_{\mathcal{A,B}}^r $, it follows $\mathrm{KD}_{\mathcal{A,B}}^{+}\cap \mathrm{span}_{\mathbb{R}}\left(\mathrm{pure}(\mathrm{KD}_{\mathcal{A,B}}^{+})\right)=\mathrm{KD}_{\mathcal{A,B}}^{+}$  . Hence, in $p^r$-dimensional space, Lemma \ref{lemma2} becomes
\begin{theorem}\label{thm1}
    Let $p$ be a prime, $r\in \mathbb{Z}^+$, for the Hilbert space of $p^r$-dimension \begin{equation}
        \mathrm{KD}_{\mathcal{A,B}}^+=\mathrm{conv}\left(\mathrm{pure}(\mathrm{KD}_{\mathcal{A,B}}^+)\right).
    \end{equation}    
\end{theorem}\par

\section{KD-classicality in Hilbert Spaces of Arbitrary Dimension via Directed Graph Path Representations}\label{any-d}

In this section, we try to characterize KD-classical states in any $d$-dimensional Hilbert spaces.

\subsection{Representation of Kirkwood-Dirac Classical Pure States}\label{msklchoose}\mbox{}\par

We first show how to represent the $m,s,k,l$ in Eq.(\ref{purekd}) in $d$-dimensional Hilbert space. Suppose $d$ has the prime factorization $d=p_1^{r_1}p_2^{r_2}...p_n^{r_n}$, where $p_u$ are primes and $p_u \neq p_v$ whenever $u\neq v$.

Consider the system of congruences for the variable $i$
\begin{equation}
  \begin{aligned}
   \left \{
   \begin{array}{ll}
    i=i_1\mod p_1^{r_1}\\
    i=i_2\mod p_2^{r_2}\\
    \vdots\\
    i=i_n\mod p_n^{r_n}\\
   \end{array}
   \right.    
  \end{aligned}.
 \end{equation}
By the Chinese Remainder Theorem, the congruence system has a unique solution $i=\sum_{u=1}^{n}i_uM_uN_u \mod d$,
where $M_u=d/p_u^{r_u}$ and $M_uN_u=1 \mod p_u^{r_u}$.
It follows that the integer $i$ can be uniquely determined by $(i_1,i_2...i_n)$, where $0\leq i_u\leq p_u^{r_u}-1$.

For each $i_u$, it can be expanded via its $p_u$-adic representation, i.e.,%%待添加向量表示的说明
\begin{equation}
    i_u=(i_{u})_{0}p_u^{0}+(i_{u})_{1}p_u^{1}+...+(i_{u})_{r_u-1}p_u^{r_u-1}=(\bm{i_u})_{0...r_u-1}
\end{equation}
where $(i_{u})_{t}\in \mathbb{Z}_{p_u}$ for any $t$. Thus $i$ can be written as
\begin{equation}\label{Crteq}
    i=\sum_u(\bm{i_u})_{0...r_u-1} M_uN_u \mod d.
\end{equation}

Since $xy=d$ in Eq.(\ref{purekd}), without loss of generality, assume that 
\begin{equation*}
    \begin{split}
       &x=p_1^{r_{(1,x)}}p_2^{r_{(2,x)}}...p_n^{r_{(n,x)}},\\
       &y=p_1^{r_{(1,y)}}p_2^{r_{(2,y)}}...p_n^{r_{(n,y)}},
    \end{split}
\end{equation*}
where $r_{(t,x)} + r_{(t,y)} = r_t$.
Since  $i=ky+m$ in Eq.(\ref{purekd}), it means $m=i \mod y$. By direct calculation we obtain $m = \sum_u\left(\bm{i_u}\right)_{0...r_{(u,y)}-1}M_{u}N_{u}\mod y$ since $M_u=d/p_u^{r_u}$ and $M_up_u^t \mod y=0$, for $\forall t\geq r_{(u,y)}$. Therefore, Eq.(\ref{Crteq}) can be written as
\begin{equation}\label{epofiwithd}
    \begin{aligned}
        i=&\sum_u\left(\bm{i_u}\right)_{0...r_{(u,y)-1}}M_uN_u
        +\sum_up_u^{r_{(u,y)}}\left( \bm{i_u}\right)_{r_{(u,y)}...r_u-1}M_uN_u\\ =&\sum_u\left(\bm{i_u}\right)_{0...r_{(u,y)-1}}M_uN_u
        +\sum_u\left(\bm{i_u} \right)_{r_{(u,y)}...r_u-1}(x/p_u^{r_{(u,x)}})N_uy.
    \end{aligned}
\end{equation}
Compared the expression for the KD-classical pure states in Eq.(\ref{purekd}), we have

\begin{equation}\label{dindexi}
\begin{aligned}
    &m=\sum_u\left(\bm{i_u}\right)_{0...r_{(u,y)}-1}M_uN_u \mod y,\\
    &k=\sum_u\left(\bm{i_u} \right)_{r_{(u,y)}...r_u-1}(x/p_u^{r_{(u,x)}})N_u \mod x.
\end{aligned}
\end{equation}
Similarly, for the subscript $j=lx+s$ of the basis state $\ket{b_j}$, we have
\begin{equation}\label{dindexj}
\begin{aligned}
    &s=\sum_v\left(\bm{j_v}\right)_{0...r_{(v,x)}-1}M_vN_v \mod x,\\
    &l=\sum_v\left(\bm{j_v} \right)_{r_{(v,x)}...r_v-1}(y/p_v^{r_{(v,y)}})N_v \mod y.
\end{aligned}
\end{equation}
Note that here we use two numbers $r_{(u,y)}$ in Eq.(\ref{dindexi}) and $r_{(v,x)}$ in Eq.(\ref{dindexj}). If $r_{(u,y)}=0$, we set $(\bm{i_u})_{0...r_{(u,y)}-1}=0$. If $r_{(u,y)}=r_u$, we set $(\bm{i_u})_{r_{(u,y)}...r_u-1}=0$. $(\bm{j_v})$ are similar.

%待修改
%By the Chinese Remainder Theorem, the residue of $i$ modulo $d$, is determined by the $p$-adic coefficients in (\ref{Crteq})
%\begin{equation*}
%    \mathbb{I}_{i}=\bigcup_u\big\{i_{(u,0)}...,i_{(u,m_u-1)}\big\},
%\end{equation*}
%Here, we denote this coefficient system as $\mathbb{I}_{i}$,  and called it the index of $i$. By the decomposition $i=ky+m$, the analogous indices for $m$ and $k$, denoted $\mathbb{I}_m$ and $\mathbb{I}_k$, are similarly defined with (\ref{dindexi}).\par

We now justify the choice of $m,k,s,l$. Specifically, we verify that these choices ensure the validity of Eq.(\ref{purekd}). To this end, it suffices to show that $k$ and $s$ take $x$ distinct values modulo $x$, and that $m$ and $l$ take $y$ distinct values modulo $y$. The complete proof is provided in Appendix \ref{value}. \par

\subsection{KD-classical States in arbitrary $d$ dimensional Hilbert space}\label{Conofd}\mbox{}\par

In this subsection, we systematically investigate the structure of KD-classical states in $d$-dimension space. We first present a generalized form of Lemma \ref{lemma1}. 

Consider two factorizations of $d$, namely $d=xy=\widetilde{x}\widetilde{y}$, where $x=\widetilde{x}p_t$ and $p_t$ is a prime factor of $d$. Without loss of generality, assume that $t=1$. Then $x=\widetilde{x}p_1$ and $\widetilde{y}=yp_1$. Hence $r_{(1,x)}-1=r_{(1,\widetilde{x})}$ and $r_{(1,y)}+1=r_{(1,\widetilde{y})}$. For $2 \leq u, v \leq n$, $r_{(v,x)}=r_{(v,\widetilde{x})}$ and $r_{(u,y)}=r_{(u,\widetilde{y})}$.

Applying Eq.(\ref{dindexi}) and Eq.(\ref{dindexj}), the expressions of $m, s, \widetilde{m}, \widetilde{s}$ under the two factorizations of $d$ are
\begin{equation}\label{m-m'}
    \begin{split}
         &m=\sum_u\left(\bm{i_u}\right)_{0...r_{(u,y)}-1}M_uN_u \;\;,\;\;\widetilde{m}=\sum_u\left(\bm{i_u}\right)_{0...r_{(u,\widetilde{y})}-1}M_uN_u,\\
         &s=\sum_v\left(\bm{j_v}\right)_{0...r_{(v,x)}-1}M_vN_v \;\;,\;\; \widetilde{s}=\sum_v\left(\bm{j_v}\right)_{0...r_{(v,\widetilde{x})}-1}M_vN_v.
    \end{split}
    \end{equation}
And $\widetilde{m}=m+(i_1)_{r_{(1,y)}}p_1^{r_{(1,y)}} M_uN_u$, $s=\widetilde{s}+(j_1)_{r_{(1,x)}-1}p_1^{r_{(1,x)}-1} M_uN_u$.

\begin{lemma}\label{dlemma1}
    
For factorizations $d = xy = \widetilde{x}\widetilde{y}$,
    \begin{equation}\label{dlemma1main}
        \sum_{(j_1)_{r_{(1,x)}-1}}\psi_{ms}=\sum_{(i_1)_{r_{(1,y)}}}\psi_{\widetilde{m}\widetilde{s}},
    \end{equation}
where $x=p_1^{r_{(1,x)}}p_2^{r_{(2,x)}}...p_n^{r_{(n,x)}}$ and $x=\widetilde{x}p_1$.     
\end{lemma}

Observing that due to the arbitrariness of $p_1$, Lemma \ref{dlemma1} remains applicable in essence as long as $x$ and ${\widetilde{x}}$ differ by exactly one prime factor $p_u$ of $d$ from each other. The complete proof of Lemma \ref{dlemma1} is given in Appendix \ref{The proof of Lemma5}.

We now define a directed graph to describe our result on KD-classical pure states. Given a factorization of $d={x_0}{y_0}$ with $\mathrm{gcd}({x_0},{y_0})=1$, without loss of generality, assume that ${x_0}=p_1^{r_1}...p_{\alpha}^{r_{\alpha}}$ and ${y_0}=p_{{\alpha}+1}^{r_{{\alpha}+1}}...p_n^{r_n}$.
Define a directed graph $G({x_0})=(V,E)$. A vertex in set $V$ corresponds to a set of KD-classical pure states which is determined by a factorization of $d$, i.e.,
\begin{equation}
    v_{x}=\bigg\{\ket{\psi_{ms}}\bra{\psi_{ms}}\bigg|\ket{\psi_{ms}}=\frac{1}{\sqrt{x}}\sum_{k=0}^{x-1}\omega^{sk}\ket{a_{k y+m}}\bigg\}
\end{equation}
with $d=xy$. Obviously, this collection $V$ consists precisely of those KD-classical pure states. The number of vertices in $V$ is equal to the number of factorizations of $d$. $E$ is the set of directed edges. There exists a directed edge between vertex $v_{x}$ and vertex $v_{\widetilde{x}}$ if and only if there exists a prime factor $p_i$ of $d$ such that $x=\widetilde{x}p_i$.
If $p_i\mid {x_0} $, the direction of the directed edge is defined to $v_{x}\rightarrow v_{\widetilde{x}}$. Otherwise, its direction is defined to $v_{\widetilde{x}}\rightarrow v_{x}$.
It is straightforward to verify that the vertex $v_{{x_0}}$ has only outgoing edges and the vertex $v_{{y_0}}$ has only incoming edges.

A path in the directed graph is an ordered sequence of vertices $(v_1,v_2...v_n)$, where for every pair of adjacent vertices $v_t$ and $v_{t+1}$, there exists a directed edge from $v_t$ to $v_{t+1}$. For a path $(v_1,v_2...v_n)$, define
\begin{equation*}
    v_{path}=\bigcup_{i=1}^nv_i.
\end{equation*}
Since each $v_i$ represents a set of KD-classical pure states,  $v_{path}$ is also a set of KD-classical pure states. The following example shows how to construct a graph.\par
\begin{example} \label{example}
For $d=2^23^3$ and ${x_0}=2^2$, the graph is\par
    
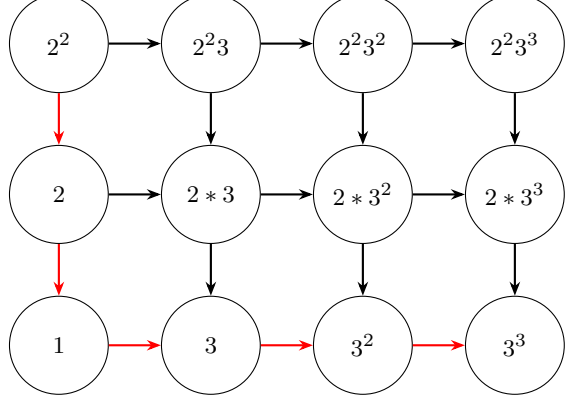
\centering
    \begin{tikzpicture}[
    > = {Stealth[scale=0.8]},     % 箭头样式
    node distance = 2.0cm,        % 节点间距
    every node/.style = {         % 节点默认样式
        circle,
        draw,                     % 描边
        minimum size = 13mm,        % 最小尺寸
        font = \small
    }]
    \node (A) at (0,0) {$2^2$};
    \node (B) [right of=A] {$2^23$};
    \node (C) [right of=B] {$2^23^2$};
    \node (D) [below of=A] {$2$};
    \node (E) [right of=D] {$2*3$};
    \node (F) [right of=E] {$2*3^2$};
    \node (G) [below of=D] {$1$};
    \node (H) [right of=G] {$3$};
    \node (I) [right of=H] {$3^2$};
    \node (J) [right of=C] {$2^23^3$};
    \node (K) [right of=F] {$2*3^3$};
    \node (L) [right of=I] {$3^3$};

    % 绘制有向边（使用 -> 表示方向）
    \draw[->,thick] (A) -- (B) ; % 添加标签
    \draw[->,thick] (B) -- (C) ;
    \draw[->,thick] (D) -- (E) ;
    \draw[->,thick] (E) -- (F) ;
    \draw[->,red,thick] (G) -- (H) ;
    \draw[->,red,thick] (H) -- (I) ;
    \draw[->,thick] (J) -- (K) ;
    \draw[->,thick] (K) -- (L) ;
    \draw[->,red,thick] (A) -- (D) ;
    \draw[->,red,thick] (D) -- (G) ;
    \draw[->,thick] (C) -- (J) ;
    \draw[->,thick] (B) -- (E) ;
    \draw[->,thick] (E) -- (H) ;
    \draw[->,thick] (F) -- (K) ;
    \draw[->,thick] (C) -- (F) ;
    \draw[->,thick] (F) -- (I) ;
    \draw[->,red,thick] (I) -- (L) ; 
\end{tikzpicture}
\end{example}

In this graph, we use red arrows to mark out a path
\begin{equation}\label{path-example}
    v_{{2^2}}\rightarrow v_{{2}}\rightarrow v_{{1}}\rightarrow v_{{3}}\rightarrow v_{{3^2}}\rightarrow v_{{3^3}}.
\end{equation}
For this path, $v_{path}$ is the union of the sets of KD-classical pure states corresponding to these vertices on the path.\par

%Recall that for $d=x_0y_0$ with $\mathrm{gcd}(x_0,y_0)=1$, we can assume that ${x_0}=p_1^{r_1}...p_{\alpha}^{r_{\alpha}}$ and ${y_0}=p_{{\alpha}+1}^{r_{{\alpha}+1}}...p_n^{r_n}$. We give the next theorem: 

\begin{theorem}\label{dthm}
Suppose $d={x_0}{y_0}$ with $\mathrm{gcd}({x_0},{y_0})=1$. For any path $v_{path}$ from $v_{x_0}$ to $v_{y_0}$ in graph $G({x_0})$,
    \begin{equation}
        \mathrm{KD}_{\mathcal{A,B}}^+\bigcap \mathrm{span}_{\mathbb{R}}(v_{path})=\mathrm{conv}(v_{path}).
    \end{equation}
\end{theorem}

In order to easily understand the proof of  $\mathrm{KD}_{\mathcal{A,B}}^+\bigcap \mathrm{span}_{\mathbb{R}}(v_{path}) \subseteq \mathrm{conv}(v_{path})$, we briefly outline the proof idea using the path of Example \ref{example}.
For a state $\rho\in \mathrm{span}_{\mathbb{R}}(v_{path})$ , $\rho$ can be expressed as a real linear combination of KD-classical states on $v_{path}$.  Using Lemma \ref{dlemma1}, the coefficients of KD-classical states in $v_{2^2}$ and $v_{2}$ are changed. Importantly, the coefficients of KD-classical states in $v_{2^2}$ become non-negative. Then applying the same method to renew the coefficients of KD-classical states in $v_{2}$ and $v_{1}$, the coefficients in $v_{2}$ non-negative. Following the same procedure, we can make all coefficients non-negative except for those in $v_{3^3}$. Finally, we use condition $\mathrm{KD}_{\mathcal{A,B}}^+$ to show that coefficients in $v_{3^3}$ is non-negative. Combined with the condition that $\rho$ is a quantum state, we establish $\rho\in\mathrm{conv}(v_{path})$.\par
The complete proof of Theorem \ref{dthm} is given in Appendix \ref{proofofthm2}.\par

Now we apply Theorem \ref{dthm} to consider Theorem \ref{thm1}. Suppose $d=p^r=x_0y_0$. We have $x_0=1$ or $x_0=p^r$ since $(x_0,y_0)=1$. Assume $x_0=1$ in Example \ref{example2}.

\begin{example} \label{example2}
For $d=p^r$ and $x_0=1$, the graph $G(x_0)$ is\par

    
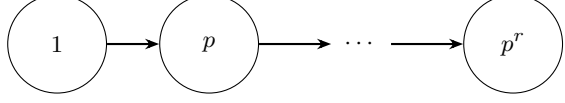
\centering
    \begin{tikzpicture}[
    > = {Stealth[scale=0.8]},     % 箭头样式
    node distance = 2.0cm,        % 增加整体间距
    every node/.style = {         % 节点默认样式
        circle,
        draw,                     % 描边
        minimum size = 13mm,      % 最小尺寸
        font = \small
    }
]
    \node (A) at (0,0) {$1$};
    \node (B) [right of=A] {$p$};
    \node (dotsBC) [right of=B, draw=none, minimum size=0] {$\cdots$};
    \node (C) [right of=dotsBC] {$p^r$};

    % 绘制有向边
    \draw[->,thick] (A) -- (B);
    \draw[->,thick] (B) -- (dotsBC);
    \draw[->,thick] (dotsBC) -- (C);
   
\end{tikzpicture}
\end{example}
There is only one path $(v_1,v_p,...,v_{p^r})$ from vertex $v_1$ to vertex $v_{p^r}$. All the KD-classical pure states are on this path, i.e., $v_{path} = \mathrm{pure}(\mathrm{KD}_{\mathcal{A,B}}^+)$.  By Theorem \ref{dthm}, we have $\mathrm{KD}_{\mathcal{A,B}}^+\bigcap \mathrm{span}_{\mathbb{R}}(v_{path})=\mathrm{conv}(v_{path})$. Since $\mathrm{KD}_{\mathcal{A,B}}^r=\mathrm{span}_{\mathbb{R}}\left(\mathrm{pure}(\mathrm{KD}_{\mathcal{A,B}}^+)\right)$ and $\mathrm{KD}_{\mathcal{A,B}}^+ \subseteq \mathrm{KD}_{\mathcal{A,B}}^r$, it follows $\mathrm{KD}_{\mathcal{A,B}}^+=\mathrm{conv}\left(\mathrm{pure}(\mathrm{KD}_{\mathcal{A,B}}^+)\right)$ that is Theorem \ref{thm1}.\par

For $d=pq$ with prime $p,q$, according to the four factorizations of $d$, there are four groups of KD-classical pure states, says $A,B,C,D$. Theorem 2 in Ref.\cite{yang2024geometry} shows that $\mathrm{KD}_{\mathcal{A,B}}^+\bigcap \mathrm{span}_{\mathbb{R}}(X \bigcup Y \bigcup Z)=\mathrm{conv}(X \bigcup Y \bigcup Z)$, where $X,Y,Z$ are any three sets of $A,B,C,D$.
The following example shows that Theorem 2 in Ref.\cite{yang2024geometry} is a special case of our Theorem \ref{dthm}.\par
\begin{example}The two graphs below are $G(x_{pq})$ and $G(x_p)$ in $pq$-dimensional space. \par

    
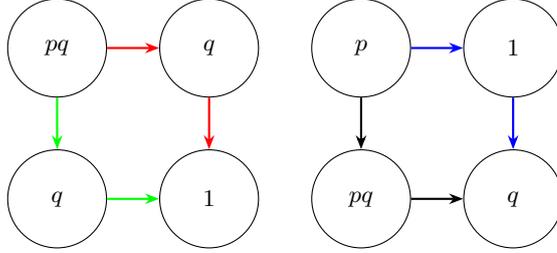
\centering
    \begin{tikzpicture}[
    > = {Stealth[scale=0.8]},     % 箭头样式
    node distance = 2.0cm,        % 节点间距
    every node/.style = {         % 节点默认样式
        circle,
        draw,                     % 描边
        minimum size = 13mm,        % 最小尺寸
        font = \small
    }
]
    \node (A) at (0,0) {$pq$};
    \node (B) [right of=A] {$q$};
    \node (C) [right of=B] {$p$};
    \node (D) [right of=C] {$1$};
    
    \node (E) [below of=A] {$q$};
    \node (F) [below of=B] {$1$};
    \node (G) [below of=C] {$pq$};
    \node (H) [below of=D] {$q$};

    % 绘制有向边（使用 -> 表示方向）
    \draw[->,red,thick] (A) -- (B) ; % 添加标签
    \draw[->,red,thick] (B) -- (F) ;
    \draw[->,green,thick] (A) -- (E) ;
    \draw[->,green,thick] (E) -- (F) ;

    \draw[->,blue,thick] (C) -- (D) ;
    \draw[->,blue,thick] (D) -- (H) ;
    \draw[->,thick] (C) -- (G) ;
    \draw[->,thick] (G) -- (H) ;
\end{tikzpicture}
\end{example}
Four different color paths, i.e., red, green, blue and black, correspond to four cases of Theorem 2 in Ref.\cite{yang2024geometry}.

\section{Conclusions and discussions}\label{conclusion-diss}

In this paper, by introducing directed graphs into the analysis of KD-quasiprobability distributions, we systematically investigate the properties of KD-classical states with the DFT matrix as the transition matrix between two given bases. Our primary results consist of two key contributions. First, we provide an alternative proof for the conclusion that the set of KD-classical states in a $p^r$-dimensional Hilbert space is the convex hull of KD-classical pure states, which validates this existing result from a new analytical perspective. Second, given that this convex hull property does not hold for Hilbert spaces of arbitrary dimension, we define a directed graph tailored to the prime factorization of the dimension $d$, and prove a general structural result: the convex hull of KD-classical pure states on any path from the start vertex to the end vertex in this graph is the intersection of the KD-classical state set and the real linear space spanned by the KD-classical pure states on that path. 
This general result is highly inclusive: it not only recovers the 
$p^r$-dimensional conclusion as a special case, but also subsumes Theorem 2 in the previous work \cite{yang2024geometry}, unifying these scattered results within a unified theoretical framework.

The study of KD-classicality is of fundamental significance for clarifying the distinction between classical and quantum physical processes. Our work provides a novel theoretical tool for investigating the structure of KD classical states. It is already known that the equality $\mathrm{KD}_{\mathcal{A,B}}^+=\mathrm{conv}(\mathrm{pure}(\mathrm{KD}_{\mathcal{A,B}}^+))$ does not hold for Hilbert spaces of arbitrary dimension, and the precise structural form of the KD-classical state set for such dimensions remains an open question. We anticipate that the directed graph approach and the general structural theorem established in this paper can serve as effective analytical tools for further exploring this unsolved problem. In addition, a natural direction for future work is to extend the analysis to KD-classical states associated with non-DFT transition matrices.

\backmatter

\bmhead{Acknowledgements}

This work is supported by the National Natural Science Foundation of China under Grant No. 12501639, Natural Science Foundation of Henan (252300420374), and Key Lab of Guangzhou for Quantum Precision Measurement under Grant No. 202201000010.

%%===================================================%%
%% For presentation purpose, we have included        %%
%% \bigskip command. Please ignore this.             %%
%%===================================================%%
\bigskip

\begin{appendices}

\section{The proof of Lemma \ref{propofKD} }\label{prop}
Recall that the KD-classical pure states have the form below
\begin{equation*}
    \ket{\psi_{ms}}=\frac{1}{\sqrt{x}}\sum_{k=0}^{x-1}\omega_x^{sk}\ket{a_{ky+m}}=\frac{1}{\sqrt{y}}\omega_d^{-ms}\sum_{l=0}^{y-1}\omega_y^{-ml}\ket{b_{lx+s}}, m\in \mathbb{Z}_y,s \in \mathbb{Z}_{x}.
\end{equation*}
Suppose $i=k'y+m'$ and $j=l'x+s'$ since $i = m' \mod y$ and $j = s' \mod x$. Then
\begin{equation*}
    \braket{a_i|\psi_{ms}}=\bra{a_{k'y+m'}}\frac{1}{\sqrt{x}}\sum_{k=0}^{x-1}\omega_x^{sk}\ket{a_{ky+m}}.
\end{equation*}
If $m\neq m'$, it is easy to check that $\braket{a_i|\psi_{ms}}=0$. If $m=m'$, $\braket{a_i|\psi_{ms}}=\frac{1}{\sqrt{x}}\omega_x^{sk'}$.
Similarly, $\braket{b_{j}|\psi_{ms'}}=\frac{1}{\sqrt{y}}\omega_d^{ms'}\omega_y^{ml'}$ if $s=s'$, and
$\braket{b_{j}|\psi_{ms'}}=0$ if $s\neq s'$.
It follows that
\begin{equation*}
    \bra{a_i}\psi_{ms}\ket{b_j}\braket{b_j|a_i}=0 \quad \text{if}\quad m\neq m'\quad \text{or} \quad s\neq s',
\end{equation*}
and
\begin{equation*}
   \bra{a_i}\psi_{ms}\ket{b_j}\braket{b_j|a_i}=\frac{1}{d}\omega_x^{s'k'}\omega_d^{m's'}\omega_y^{m'l'}\omega_{d}^{-ij}=\frac{1}{d}
\end{equation*}
if $m= m'$ and $s=s'$.
Here, we use
\begin{equation*}
    \omega_{d}^{ij}=\omega_{d}^{(k'y+m')(l'x+s')}=\omega_{d}^{m'l'x+k's'y+m's'}=\omega_x^{s'k'}\omega_d^{m's'}\omega_y^{m'l'}.
\end{equation*}
That completes the proof.

%%==========================================%%

\section{The proof of Lemma \ref{lemma1}}
\label{The proof of Lemma1}
\begin{proof}
        For the left-hand side of the equality, from Eq.(\ref{p_use}), we obtain        
        \begin{equation*}
        \begin{aligned}  
            \psi_{\substack{\bm{i}_{0...\beta-1}\\\bm{j}_{0...\alpha-1}}}
            =&\frac{1}{p^{\alpha}}\sum_{\substack{\bm{i}_{\beta...r-1}\\\bm{i}'_{\beta...r-1}}}
            \omega_{p^{\alpha}}^{\Delta^{\bm{i}\bm{i'}}_{\beta...r-1}\bm{j}_{0...\alpha-1}}
            \ket{a_{\bm{i}_{0...\beta-1}+p^{\beta}\bm{i}_{\beta...r-1}}}
            \bra{a_{\bm{i}_{0...\beta-1}+p^{\beta}\bm{i}'_{\beta...r-1}}}
        \end{aligned}
        \end{equation*}
        Note that here we use the notation $\Delta^{\bm{a}\bm{a}'}_{\beta...r-1}=\bm{a}_{\beta...r-1}-\bm{a}'_{\beta...r-1}$. By direct calculation, we have
        \begin{equation*}
            \omega_{p^{\alpha}}^{\Delta^{\bm{i}\bm{i'}}_{\beta...r-1}\bm{j}_{0...\alpha-1}}
            =\omega_{p^{\alpha}}^{\Delta^{\bm{i}\bm{i'}}_{\beta...r-1}\bm{j}_{0...\alpha-2}}
            \omega_{p^{\alpha}}^{\Delta^{\bm{i}\bm{i'}}_{\beta}j_{\alpha-1}p^{\alpha-1}}
            \omega_{p^{\alpha}}^{\Delta^{\bm{i}\bm{i'}}_{\beta+1...r-1}j_{\alpha-1}p^{\alpha}}            
            =\omega_{p^{\alpha}}^{\Delta^{\bm{i}\bm{i'}}_{\beta...r-1}\bm{j}_{0...\alpha-2}}
            \omega_{p^{\alpha}}^{\Delta^{\bm{i}\bm{i'}}_{\beta}j_{\alpha-1}p^{\alpha-1}}.
        \end{equation*}
%        For simplicity, we denote $\omega_{p^{\alpha}}^{\Delta^{\bm{i}\bm{i'}}_{\beta}j_{\alpha-1}p^{\alpha-1}} \omega_{p^{\alpha}}^{\Delta^{\bm{i}\bm{i'}}_{\beta...r-1}\bm{j}_{0...\alpha-2}}
%        \ket{a_{\bm{i}_{0...\beta-1}+p^{\beta}\bm{i}_{\beta...r-1}}}
%            \bra{a_{\bm{i}_{0...\beta-1}+p^{\beta}\bm{i}'_{\beta...r-1}}}$ by 
%            $A_{\substack{\bm{i}_{\beta...r-1}\bm{i}'_{\beta...r-1}}}$.
        Thus
            \begin{equation}\label{doperator1}
            \begin{aligned}  
                 \psi_{\substack{\bm{i}_{0...\beta-1}\\\bm{j}_{0...\alpha-1}}}
                 =&\frac{1}{p^{\alpha}}\sum_{\substack{\bm{i}_{\beta...r-1}\\\bm{i}'_{\beta...r-1}}}
                 \omega_{p^{\alpha}}^{\Delta^{\bm{i}\bm{i'}}_{\beta...r-1}\bm{j}_{0...\alpha-2}}
                 \omega_{p^{\alpha}}^{\Delta^{\bm{i}\bm{i'}}_{\beta}j_{\alpha-1}p^{\alpha-1}} 
        \ket{a_{\bm{i}_{0...\beta-1}+p^{\beta}\bm{i}_{\beta...r-1}}}
            \bra{a_{\bm{i}_{0...\beta-1}+p^{\beta}\bm{i}'_{\beta...r-1}}}\\
            :=& \frac{1}{p^{\alpha}}\sum_{\substack{\bm{i}_{\beta...r-1}\\\bm{i}'_{\beta...r-1}}}
            A_{\substack{\bm{i}_{\beta...r-1} \\ \bm{i}'_{\beta...r-1}}}
            \end{aligned}
            \end{equation}
        
        The summation over $j_{\alpha-1}$ is in Eq.(\ref{doperator1}), i.e.,
        \begin{equation}\label{over-j}
            \sum_{j_{\alpha-1}}\psi_{\substack{\bm{i}_{0...\beta-1}\\\bm{j}_{0...\alpha-1}}}
            =\frac{1}{p^{\alpha}}\sum_{j_{\alpha-1}}\sum_{\substack{\bm{i}_{\beta...r-1}\\\bm{i}'_{\beta...r-1}}}A_{\substack{\bm{i}_{\beta...r-1}\\\bm{i}'_{\beta...r-1}}}
            =\frac{1}{p^{\alpha}}\sum_{j_{\alpha-1}}\big\{\sum_{i_{\beta}=i_{\beta}'}+\sum_{i_{\beta}\neq i_{\beta}'}\big\}\sum_{\substack{\bm{i}_{\beta+1...r-1}\\\bm{i}'_{\beta+1...r-1}}}A_{\substack{\bm{i}_{\beta...r-1}\\\bm{i}'_{\beta...r-1}}}.
        \end{equation}
        Note that the 2nd equality follows from $\sum_{\bm{i}_{\beta...r-1}}=\sum_{i_{\beta}}\sum_{\bm{i}_{\beta+1...r-1}}$
        since $\bm{i}_{\beta...r-1}=i_{\beta}+p\bm{i}_{\beta+1...r-1}$.

        We now discuss the two cases $i_{\beta}=i_{\beta}'$ and $i_{\beta}\neq i_{\beta}'$, respectively.
        
        \textbf{Case 1}. $i_{\beta}=i_{\beta}'$, i.e., $\Delta^{\bm{i}\bm{i'}}_{\beta}=0$. The coefficient in Eq.(\ref{doperator1}) can be simplified to 
        \begin{equation*}
            \omega_{p^{\alpha}}^{\Delta^{\bm{i}\bm{i'}}_{\beta...r-1}\bm{j}_{0...\alpha-2}}
        =\omega_{p^{\alpha}}^{(\Delta^{\bm{i}\bm{i'}}_{\beta}+p\Delta^{\bm{i}\bm{i'}}_{\beta+1...r-1})\bm{j}_{0...\alpha-2}}
        =\omega_{p^{\alpha}}^{p\Delta^{\bm{i}\bm{i'}}_{\beta+1...r-1}\bm{j}_{0...\alpha-2}}
        =\omega_{p^{\alpha-1}}^{\Delta^{\bm{i}\bm{i'}}_{\beta+1...r-1}\bm{j}_{0...\alpha-2}}.
        \end{equation*}
        Notice that the simplified coefficient is independent of $j_{\alpha-1}$ and 
        $a_{\bm{i}_{0...\beta-1}+p^{\beta}\bm{i}'_{\beta...r-1}}=a_{\bm{i}_{0...\beta}+p^{\beta+1}\bm{i}'_{\beta+1...r-1}}$. Therefore, 
        \begin{equation}
        \begin{aligned}
        &\frac{1}{p^{\alpha}}\sum_{j_{\alpha-1}}\sum_{i_{\beta}=i_{\beta}'}
        \sum_{\substack{\bm{i}_{\beta+1...r-1}\\\bm{i}'_{\beta+1...r-1}}}A_{\substack{\bm{i}_{\beta...r-1}\\\bm{i}'_{\beta...r-1}}}\\
%            =&\frac{1}{p^{\alpha}}\sum_{j_{\alpha-1}}\sum_{i_{\beta}=i_{\beta}'}
%            \sum_{\substack{\bm{i}_{\beta+1...r-1}\\\bm{i}'_{\beta+1...r-1}}}
%            \omega_{p^{\alpha}}^{\Delta^{\bm{i}\bm{i'}}_{\beta...r-1}\bm{j}_{0...\alpha-2}}\ket{a_{\bm{i}_{\beta...r-1}p^{\beta}+\bm{i}_{0...\beta-1}}}
%            \bra{a_{\bm{i}'_{\beta...r-1}p^{\beta}+\bm{i}_{0...\beta-1}}}\\
%            =&\frac{1}{p^{\alpha}}\sum_{j_{\alpha-1}}\sum_{i_{\beta}=i_{\beta}'}
%            \sum_{\substack{\bm{i}_{\beta+1...r-1}\\\bm{i}'_{\beta+1...r-1}}}
%         \omega_{p^{\alpha}}^{(\Delta^{\bm{i}\bm{i'}}_{\beta}+p\Delta^{\bm{i}\bm{i'}}_{\beta+1...r-1})\bm{j}_{0...\alpha-2}}
%            \ket{a_{\bm{i}_{\beta...r-1}p^{\beta}+\bm{i}_{0...\beta-1}}}
%            \bra{a_{\bm{i}'_{\beta...r-1}p^{\beta}+\bm{i}_{0...\beta-1}}}\\
           =&\frac{1}{p^{\alpha-1}}\sum_{i_{\beta}}
           \sum_{\substack{\bm{i}_{\beta+1...r-1}\\\bm{i}'_{\beta+1...r-1}}}
           \omega_{p^{\alpha-1}}^{\Delta^{\bm{i}\bm{i'}}_{\beta+1...r-1}\bm{j}_{0...\alpha-2}}
           \ket{a_{\bm{i}_{0...\beta}+p^{\beta+1}\bm{i}_{\beta+1...r-1}}}
            \bra{a_{\bm{i}_{0...\beta}+p^{\beta+1}\bm{i}'_{\beta+1...r-1}}}\\ 
           %\omega_{p^{t_1}-1}^{(i_{t_2+1}-i_{t_2+1}',...i_{M-1}-i_{M-1}')(j_0...j_{t_1-2})}\\
           %&\times\ket{a_{(i_{t_2}...i_{M-1})p^{t_2}+(i_0...i_{t_2-1})}}
           %\bra{a_{(i_{t_2}'...i_{M-1}')p^{t_2}+(i_0...i_{t_2-1})}}.  
        \end{aligned}
        \end{equation}
       
       \textbf{Case 2} . $i_{\beta}\ne i_{\beta}'$, i.e., $\Delta^{\bm{i}\bm{i'}}_{\beta}\neq0$.
       
       The coefficient in Eq.(\ref{doperator1}) has the following property when $\Delta^{\bm{i}\bm{i'}}_{\beta}\neq0$,
       \begin{equation*}
           \sum_{j_{\alpha-1}} 
       \omega_{p^{\alpha}}^{\Delta^{\bm{i}\bm{i'}}_{\beta...r-1}\bm{j}_{0...\alpha-2}}
       \omega_{p^{\alpha}}^{\Delta^{\bm{i}\bm{i'}}_{\beta}j_{\alpha-1}p^{\alpha-1}}
       =\omega_{p^{\alpha}}^{\Delta^{\bm{i}\bm{i'}}_{\beta...r-1}\bm{j}_{0...\alpha-2}}
       \sum_{j_{\alpha-1}} \omega_{p}^{\Delta^{\bm{i}\bm{i'}}_{\beta}j_{\alpha-1}}=0.
       \end{equation*}
       Thus
       \begin{equation}\label{r}
           \begin{aligned}
               \frac{1}{p^{\alpha}}\sum_{j_{\alpha-1}}\sum_{i_{\beta}\neq i_{\beta}'}
               \sum_{\substack{\bm{i}_{\beta+1...r-1}\\\bm{i}'_{\beta+1...r-1}}}A_{\substack{\bm{i}_{\beta...r-1}\\\bm{i}'_{\beta...r-1}}}=0.
            %\omega_{p^{t_1}}^{(i_{t_2}-i_{t_2}')j_{t_1-1}^*p^{t_1-1}}\\
            %&\times\omega_{p^{t_1}}^{(i_{t_2}-i_{t_2}'...i_{M-1}-i_{M-1}')(j_0...j_{t_1-2})}\\
            %&\times\ket{a_{(i_{t_2}...i_{M-1})p^{t_2}+(i_0...i_{t_2-1})}}
            %\bra{a_{(i_{t_2}'...i_{M-1}')p^{t_2}+(i_0...i_{t_2-1})}}\\
%            =&\frac{1}{p^{\alpha}}\sum_{\substack{\bm{i}_{\beta+1...r-1}\\\bm{i}'_{\beta+1...r-1}}}\sum_{i_{\beta}\neq i_{\beta}'}
%            \big\{\sum_{j_{\alpha-1}}\omega_{p^{\alpha}}^{\Delta^{\bm{i}\bm{i'}}_{\beta}j_{\alpha-1}p^{\alpha-1}}\big\}\\
%            &\times\omega_{p^{\alpha}}^{\Delta^{\bm{i}\bm{i'}}_{\beta...r-1}\bm{j}_{0...\alpha-2}}\\   &\times\ket{a_{\bm{i}_{\beta...r-1}p^{\beta}+\bm{i}_{0...\beta-1}}}
%            \bra{a_{\bm{i}'_{\beta...r-1}p^{\beta}+\bm{i}_{0...\beta-1}}}\\ 
%            =&0.
           \end{aligned}
       \end{equation}
%    The last equality in equation (\ref{r}) holds because $i_{\beta}\ne i_{\beta}'$ and $\omega_{p^{\alpha}}$ is a $p^{\alpha}$-th root of unity, thereby leading to..
       
%       \begin{equation}
%            \{\sum_{j_{\alpha-1}}\omega_{p^{\alpha}}^{\Delta^{\bm{i}\bm{i'}}_{\beta}j_{\alpha-1}p^{\alpha-1}}\}=
%            \{\sum_{j_{\alpha-1}}\omega_{p}^{\Delta^{\bm{i}\bm{i'}}_{\beta}j_{\alpha-1}}\}=0.
%        \end{equation}

By the discussion in Case 1 and Case 2, Eq.(\ref{over-j}) can be written as
\begin{equation}\label{lemmarleft}
    \begin{aligned}
        \sum_{j_{\alpha-1}}\psi_{\substack{\bm{i}_{0...\beta-1}\\\bm{j}_{0...\alpha-1}}}
        =&\frac{1}{p^{\alpha-1}}\sum_{i_{\beta}}
        \sum_{\substack{\bm{i}_{\beta+1...r-1}\\\bm{i}_{\beta+1...r-1}'}}
           \omega_{p^{\alpha}-1}^{\Delta^{\bm{i}\bm{i'}}_{\beta+
           1...r-1}\bm{j}_{0...\alpha-2}}
           \ket{a_{\bm{i}_{0...\beta}+p^{\beta+1}\bm{i}_{\beta+1...r-1}}}
            \bra{a_{\bm{i}_{0...\beta}+p^{\beta+1}\bm{i}'_{\beta+1...r-1}}}.
    \end{aligned}
\end{equation}

For the right-hand side of the equality, it follows from Eq.(\ref{p_use}) that
\begin{equation} \label{right-1}
    \begin{aligned}
        \psi_{\substack{\bm{i}_{0...\beta}\\\bm{j}_{0...\alpha-2}}}    
        =&\frac{1}{p^{\alpha-1}}
        \sum_{\substack{\bm{i}_{\beta+1...r-1}\\\bm{i}'_{\beta+1...r-1}}}
         \omega_{p^{\alpha-1}}^{\Delta^{\bm{i}\bm{i'}}_{\beta+1...r-1}\bm{j}_{0...\alpha-2}}
         \ket{a_{\bm{i}_{0...\beta}+p^{\beta+1}\bm{i}_{\beta+1...r-1}}}
            \bra{a_{\bm{i}_{0...\beta}+p^{\beta+1}\bm{i}'_{\beta+1...r-1}}}.\\
    \end{aligned}
\end{equation}
When the summation over $i_{\beta}$ is in Eq.(\ref{right-1}), by comparing Eq.(\ref{right-1}) and Eq.(\ref{lemmarleft}) one can find $\sum_{i_{\beta}}\psi_{\substack{\bm{i}_{0...\beta}\\\bm{j}_{0...\alpha-2}}} = \sum_{j_{\alpha-1}}\psi_{\substack{\bm{i}_{0...\beta-1}\\\bm{j}_{0...\alpha-1}}}$.\par
      
\end{proof}

%%================================================%%

\section{The proof of Lemma \ref{lemma2}}\label{The proof of Lemma2}

\begin{proof}
It is straightforward to verify the two inclusion relations
\begin{equation*}
    \begin{aligned}
        &\mathrm{conv}\left(\mathrm{pure}(\mathrm{KD}_{\mathcal{A,B}}^{+})\right)\subseteq \mathrm{KD}_{\mathcal{A,B}}^{+},\\
        &\mathrm{conv}\left(\mathrm{pure}(\mathrm{KD}_{\mathcal{A,B}}^{+})\right)\subseteq \mathrm{span}_{\mathbb{R}}\left(\mathrm{pure}(\mathrm{KD}_{\mathcal{A,B}}^{+})\right).
    \end{aligned}
\end{equation*}
These two inclusion relations imply that

\begin{equation*}\mathrm{conv}\left(\mathrm{pure}(\mathrm{KD}_{\mathcal{A,B}}^{+})\right)\subseteq \mathrm{KD}_{\mathcal{A,B}}^{+}\cap \mathrm{span}_{\mathbb{R}}\left(\mathrm{pure}(\mathrm{KD}_{\mathcal{A,B}}^{+})\right).
\end{equation*}

Now we prove the converse inclusion.
For any $\rho\in \mathrm{KD}_{\mathcal{A,B}}^{+}\cap \mathrm{span}_{\mathbb{R}}\left(\mathrm{pure}(\mathrm{KD}_{\mathcal{A,B}}^{+})\right)$, we have $\rho\in\mathrm{span}_{\mathbb{R}}\left(\mathrm{pure}(\mathrm{KD}_{\mathcal{A,B}}^{+})\right)$. It means $\rho$ can be expressed as a real linear combination of KD-classical pure states:
\begin{equation}\label{inspanrkd}
    \begin{aligned} 
     &\rho
    =\sum_{\bm{i}_{0...r-1}}
    \lambda_{\bm{i}_{0...r-1}}a_{\bm{i}_{0...r-1}}
    +\sum_{\alpha=1}^{r-1}\sum_{\substack{\bm{i}_{0...\beta-1}\\\bm{j}_{0...\alpha-1}}}
    \lambda_{\substack{\bm{i}_{0...\beta-1}\\\bm{j}_{0...\alpha-1}}}\psi_{\substack{\bm{i}_{0...\beta-1}\\\bm{j}_{0...\alpha-1}}}
    +\sum_{\bm{j}_{0...r-1}}
    \lambda_{\bm{j}_{0...r-1}}b_{\bm{j}_{0...r-1}}.
    \end{aligned}
\end{equation}
where $\beta+\alpha=r$.

We need to prove that for a new expression of $\rho$, all the new coefficients of $\rho$ are nonnegative.
We adopt a two-step approach to complete this proof.

\textbf{Step 1}. In this step, we construct the new expression of $\rho$ in Eq.(\ref{inspanrkd}) and show that the components of the first $r$ groups are non-negative.\par

We first consider the 1st and 2nd sums. For $\bm{i}_{0...r-1}=\bm{i}_{0...r-2}+p^{r-1}i_{r-1}$, we rewrite the sum $\sum_{\bm{i}_{0...r-1}}=\sum_{\bm{i}_{0...r-2}}\sum_{\bm{i}_{r-1}}$. With this notation, we rewrite the 1st and 2nd groups of $\rho$
\begin{equation*}
    \begin{aligned}
    \rho&
    =\sum_{\bm{i}_{0...r-2}}\sum_{i_{r-1}}\lambda_{(\bm{i}_{0...r-1})}a_{\bm{i}_{0...r-1}}+\sum_{\bm{i}_{0...r-2}}\sum_{j_0}\lambda_{\substack{\bm{i}_{0...r-2}\\j_0}}\psi_{\substack{\bm{i}_{0...r-1}\\j_0}}+...\\
    &=\sum_{\bm{i}_{0...r-2}}\sum_{i_{r-1}}\left(\lambda_{\bm{i}_{0...r-1}}-\lambda_{\bm{i}_{0...r-2}}^{\mathrm{min}}\right)
    a_{\bm{i}_{0...r-1}}
    +\sum_{\bm{i}_{0...r-2}}\sum_{j_0}\left(\lambda_{\substack{\bm{i}_{0...r-2}\\j_0}}+\lambda_{\bm{i}_{0...r-2}}^{\mathrm{min}}\right)\psi_{\substack{\bm{i}_{0...r-2}\\j_0}}+...\\
    &:=\sum_{\bm{i}_{0...r-2}}\sum_{i_{r-1}}\left(\lambda_{\bm{i}_{0...r-1}}-\lambda_{\bm{i}_{0...r-2}}^{\mathrm{min}}\right)
    a_{\bm{i}_{0...r-1}}
    +\sum_{\bm{i}_{0...r-2}}\sum_{j_0}\widetilde\lambda_{\substack{\bm{i}_{0...r-2}\\j_0}}\psi_{\substack{\bm{i}_{0...r-2}\\j_0}}+...
    \end{aligned}
\end{equation*}
where $\lambda_{\bm{i}_{0...r-2}}^{\mathrm{min}}=\min_{i_{r-1}}\left\{\lambda_{\bm{i}_{0...r-1}}\right\}$ . Note that $\lambda_{\bm{i}_{0...r-2}}^{\mathrm{min}}$ does not depend on $i_{r-1}$, but only on $\bm{i}_{0...r-2}$. Here, for simplicity, we denote $\left(\lambda_{\substack{\bm{i}_{0...r-2}\\j_0}}+\lambda_{\bm{i}_{0...r-2}}^{\mathrm{min}}\right)$ by $\widetilde\lambda_{\substack{\bm{i}_{0...r-2}\\j_0}}$. 
The 2nd equality holds since  $\sum_{i_{r-1}}a_{\bm{i}_{0...r-1}}=
    \sum_{j_0}\psi_{\substack{\bm{i}_{0...r-2}\\j_0}}$ in Corollary \ref{corollary}. 
    Now the new coefficients of the first group of $\rho$ are nonnegative.

%The preceding formula is equivalent to

%\begin{equation*}
%    \Gamma\left(\{\lambda_{\bm{i}_{0...r-1}}\},\{\lambda_{\substack{\bm{i}_{0...r-2}\\j_0}}\}...\right)
%    =\Gamma\left(\{\lambda_{\bm{i}_{0...r-1}}-\lambda_{\bm{i}_{0...r-2}}^{\mathrm{min}}\},\{\widetilde\lambda_{\substack{\bm{i}_{0...r-2}\\j_0}}\}...\right).
%\end{equation*}
For the 2nd to the $r$-th group of coefficients, 
%noted that the sum with subscript $\bm{j}_{0...\alpha}$ is corresponding to the $(\alpha+2)$-th sum of Eq.(\ref{inspanrkd}). 
we can employ Lemma \ref{lemma1} to construct the nonnegative coefficient for $1\leq \alpha\leq r-2$ and $\beta+\alpha=r$ in Eq.(\ref{inspanrkd}). Generally, consider the coefficients of the following two groups,
%for the $\alpha+1$ group and change the coefficient of $\alpha+2$ group with $1\leq \alpha\leq r-2$ and $\beta+\alpha=r$
\begin{equation*}
\begin{split}
   &\sum_{\substack{\bm{i}_{0...\beta-1}\\ \bm{j}_{0...\alpha-1}}}\widetilde\lambda_{\substack{\bm{i}_{0...\beta-1}\\ \bm{j}_{0...\alpha-1}}}\psi_{\substack{\bm{i}_{0...\beta-1}\\ \bm{j}_{0...\alpha-1}}}+\sum_{\substack{\bm{i}_{0...\beta-2}\\ \bm{j}_{0...\alpha}}}\lambda_{\substack{\bm{i}_{0...\beta-2}\\ \bm{j}_{0...\alpha}}}\psi_{\substack{\bm{i}_{0...\beta-2}\\ \bm{j}_{0...\alpha}}}\\
   =&\sum_{\substack{\bm{i}_{0...\beta-2}\\ \bm{j}_{0...\alpha-1}}}\sum_{\bm{i}_{\beta-1}}\left(\widetilde\lambda_{\substack{\bm{i}_{0...\beta-1}\\ \bm{j}_{0...\alpha-1}}}-\widetilde{\lambda}^{\min}_{\substack{\bm{i}_{0...\beta-2}\\ \bm{j}_{0...\alpha-1}}}\right)\psi_{\substack{\bm{i}_{0...\beta-1}\\ \bm{j}_{0...\alpha-1}}}+\sum_{\substack{\bm{i}_{0...\beta-2}\\ \bm{j}_{0...\alpha-1}}}\sum_{\bm{j}_{\alpha}}\left(\lambda_{\substack{\bm{i}_{0...\beta-2}\\ \bm{j}_{0...\alpha}}}+\widetilde{\lambda}^{\min}_{\substack{\bm{i}_{0...\beta-2}\\ \bm{j}_{0...\alpha-1}}}\right)\psi_{\substack{\bm{i}_{0...\beta-2}\\ \bm{j}_{0...\alpha}}}\\
   :=&\sum_{\substack{\bm{i}_{0...\beta-2}\\ \bm{j}_{0...\alpha-1}}}\sum_{\bm{i}_{\beta-1}}\left(\widetilde\lambda_{\substack{\bm{i}_{0...\beta-1}\\ \bm{j}_{0...\alpha-1}}}-\widetilde{\lambda}^{\min}_{\substack{\bm{i}_{0...\beta-2}\\ \bm{j}_{0...\alpha-1}}}\right)\psi_{\substack{\bm{i}_{0...\beta-1}\\ \bm{j}_{0...\alpha-1}}}+\sum_{\substack{\bm{i}_{0...\beta-2}\\ \bm{j}_{0...\alpha-1}}}\sum_{\bm{j}_{\alpha}}\widetilde\lambda_{\substack{\bm{i}_{0...\beta-2}\\\bm{j}_{0...\alpha}}}\psi_{\substack{\bm{i}_{0...\beta-2}\\ \bm{j}_{0...\alpha}}}
\end{split}
\end{equation*}
where $\widetilde\lambda_{\substack{\bm{i}_{0...\beta-2}\\\bm{j}_{0...\alpha-1}}}^{\mathrm{min}}=\min\limits_{{i_{\beta-1}}}\{\widetilde\lambda_{\substack{\bm{i}_{0...\beta-1}\\\bm{j}_{0...\alpha-1}}}\} $ with $\widetilde\lambda_{\substack{\bm{i}_{0...\beta-2}\\\bm{j}_{0...\alpha-1}
    }}^{\mathrm{min}}$ dose not depend on $i_{\beta-1}$ but only on $\bm{i}_{0...\beta-2}$ and $\bm{j}_{0...\alpha-1}$. We denoted $\widetilde\lambda_{\substack{\bm{i}_{0...\beta-2}\\\bm{j}_{0...\alpha}}}
    =\lambda_{\substack{\bm{i}_{0...\beta-2}\\\bm{j}_{0...\alpha}}}+\widetilde{\lambda}_{\substack{\bm{i}_{0...\beta-2}\\\bm{j}_{0...\alpha-1}
    }}^{\min} $. The sum notation change since $\bm{i}_{0...\beta-1}=\bm{i}_{0...\beta-2}+\bm{i}_{\beta-1}p^{\beta-1}$ and $\bm{j}_{0...\alpha}=\bm{j}_{0...\alpha-1}+\bm{j}_{\alpha}p^{\alpha}$. \par
    
We have constructed the nonnegative coefficients for the 1st sum and $(r-2)$-th sum in Eq.(\ref{inspanrkd}), and make the coefficients of $(r-1)$-th sum $\widetilde\lambda_{\substack{i_0\\\bm{j}_{0...r-2}}}$. By the same way construct nonnegative coefficient of $(r-1)$-th sum and obtain the new expression of $\rho$:
\begin{equation}\label{newexpression}
    \begin{aligned}
        \rho
    &=\sum_{\bm{i}_{0...r-1}}
    \left(\lambda_{\bm{i}_{0...r-1}}-\lambda_{\bm{i}_{0...r-2}}^{\min}\right)
    a_{\bm{i}_{0...r-1}}
    +\sum_{\alpha=1}^{r-2}\sum_{\substack{\bm{i}_{0...\beta-1}\\\bm{j}_{0...\alpha-1}}}
    \left(\widetilde\lambda_{\substack{\bm{i}_{0...\beta-1}\\\bm{j}_{0...\alpha-1}}}
    -\widetilde\lambda_{\substack{\bm{i}_{0...\beta-2}\\\bm{j}_{0...\alpha-1}}}^{\min}\right)
    \psi_{\substack{\bm{i}_{0...\beta-1}\\\bm{j}_{0...\alpha-1}}}\\
    &+\sum_{\substack{i_0\\\bm{j}_{0...r-2}}}\left(\widetilde\lambda_{\substack{i_0\\\bm{j}_{0...r-2}}}-\widetilde\lambda_{\bm{j}_{0...r-2}}^{\min}\right)\psi_{\substack{i_0\\\bm{j}_{0...r-2}}}+\sum_{\bm{j}_{0...r-1}}
    \left(\lambda_{\bm{j}_{0...r-1}}+\widetilde\lambda_{\bm{j}_{0...r-2}}^{\min}\right)
    b_{\bm{j}_{0...r-1}}
    .
    \end{aligned}
\end{equation}
Now all the coefficients are nonnegative except for $\lambda_{\bm{j}_{0...r-1}}+\widetilde\lambda_{\bm{j}_{0...r-2}}^{\min}$. In step 2, we will show that it is nonnegative.

\textbf{Step 2}. For any index $j'=\bm{j'}_{0...r-1}$, we will expand $\lambda_{\bm{j'}_{0...r-1}}+\widetilde\lambda_{\bm{j'}_{0...r-2}}^{\mathrm{min}}$ as the sum of the original coefficient $\lambda$ and prove that the sum is nonnegative. 

Since $\widetilde\lambda_{\bm{j'}_{0...r-2}}^{\mathrm{min}}
        =\min_{i_0}\{\widetilde\lambda_{\substack{i_0\\\bm{j'}_{0...r-2}}}\}$,
it follows that there exists $i_0^*$ such that $\min_{i_0}\{\widetilde\lambda_{\substack{i_0\\\bm{j'}_{0...r-2}}}\}=
\widetilde\lambda_{\substack{i_0^*\\\bm{j'}_{0...r-2}}}$.
By the definition of $\widetilde\lambda$,
\begin{equation}\label{lamada-expand}
    \quad\lambda_{\bm{j'}_{0...r-1}}+\widetilde\lambda_{\bm{j'}_{0...r-2}}^{\min}=
        \lambda_{\bm{j'}_{0...r-1}}+\widetilde\lambda_{\substack{i_0^*\\\bm{j'}_{0...r-2}}}\\
        =\lambda_{\bm{j'}_{0...r-1}}+\lambda_{\substack{i_0^*\\\bm{j'}_{0...r-2}}}
        +\widetilde\lambda_{\substack{i_0^*\\\bm{j'}_{0...r-3}}}^{\min}.
\end{equation}
Similarly, for each $\widetilde\lambda_{\substack{\bm{i^*}_{0...\beta-2}\\\bm{j'}_{0...\alpha-1}}}^{\min}$ with $ r-1\geq\alpha\geq 1$ and $\beta+\alpha=r$, there exists $i_{\beta-1}^*$ such that
\begin{equation*}
\widetilde\lambda_{\substack{\bm{i^*}_{0...\beta-2}\\\bm{j'}_{0...\alpha-1}}}^{\min}=\widetilde\lambda_{\substack{\bm{i^*}_{0...\beta-2}+p^{\beta-1}i^*_{\beta-1}\\\bm{j'}_{0...\alpha-1}}}:=\widetilde\lambda_{\substack{\bm{i^*}_{0...\beta-1}\\ \bm{j'}_{0...\alpha-1}}}=\lambda_{\substack{\bm{i^*}_{0...\beta-1}\\ \bm{j'}_{0...\alpha-1}}}+\widetilde\lambda_{\substack{\bm{i^*}_{0...\beta-1}\\\bm{j'}_{0...\alpha-2}}}^{\min}.
\end{equation*}
Here we denote $\bm{i}^*_{0...\beta-1}=\bm{i}^*_{0...\beta-2}+p^{\beta-1}i^*_{\beta-1}$. Repeatedly applying this recurrence relation, Eq.(\ref{lamada-expand}) can be written as
\begin{equation}\label{finalcoefficent}
    \lambda_{\bm{j'}_{0...r-1}}+\widetilde\lambda_{\bm{j'}_{0...r-2}}^{\min}=\lambda_{\bm{j'}_{0...r-1}}+
    \sum_{\substack{r-1\geq \alpha \geq1\\ \alpha+\beta=r}}\lambda_{\substack{\bm{i^*}_{0...\beta-1}\\\bm{j'}_{0...\alpha-1}}}+\lambda_{\bm{i^*}_{0...r-1}}.
\end{equation}

We now prove that Eq.(\ref{finalcoefficent}) is nonnegative.
For subscripts $j'=\bm{j'}_{0...r-1}$ and $i^*=\bm{i^*}_{0...r-1}$, we use Lemma \ref{propofKD} to calculate the KD-distribution with the basis states $\ket{a_{i^*}},\ket{b_{j'}}$:
\begin{equation}\label{fingeq}
\begin{aligned}
    \bra{a_{i^*}}\rho\ket{b_{j'}}\braket{b_{j'}|a_{i^*}}
        =|\braket{a_{i^*}|b_{j'}}|^2\left(\lambda_{\bm{i^*}_{0...r-1}}
        +\sum_{\substack{r-1\geq \alpha\geq1\\ \alpha+\beta=r}}\lambda_{\substack{\bm{i^*}_{0...\beta-1}\\\bm{j'}_{0...\alpha-1}}}
        +\lambda_{\bm{j'}_{0...r-1}}\right) \geq 0,    
\end{aligned}
\end{equation}
since $\rho\in \mathrm{KD}_{\mathcal{A,B}}^+$.
From the property of the DFT matrix, $|\braket{a_{i^*}|b_{j'}}|^2=\frac{1}{d}$, we have 
\begin{equation*}
    \lambda_{\bm{i^*}_{0...r-1}}
        +\sum_{\substack{r-1\geq \alpha\geq1\\ \alpha+\beta=r}}\lambda_{\substack{\bm{i^*}_{0...\beta-1}\\\bm{j'}_{0...\alpha-1}}}
        +\lambda_{\bm{j'}_{0...r-1}} \geq 0.
\end{equation*}
That is, Eq.(\ref{finalcoefficent}) is nonnegative.

Due to the arbitrariness of $j'$, all coefficients in the final group are non-negative. Furthermore, all the coefficients of Eq.(\ref{newexpression}) are nonnegative. 

Finally, $\mathrm{Tr}(\rho)=1$ implies that the sum of all the coefficients of Eq.(\ref{newexpression}) is 1.
Therefore, $\rho\in \mathrm{conv}\left(\mathrm{pure(}\mathrm{KD}_{\mathcal{A,B}}^+)\right)$, which is the desired result.
\end{proof}

%%=============================================%%

\section{The value range of $k$ in Eq.(\ref{dindexi}) }
\label{value}

\begin{proof}
    Suppose that the number of distinct residues of $k$ in $\mathbb{Z}_{x}$ is $t$ with $t<x$. Next, we consider
\begin{equation}
    m=\sum_{u=1}^{n}(\bm{i_u})_{0...r_{(u,y)}}M_uN_u.
\end{equation}
Given that each $(i_u)_{t}\in \mathbb{Z}_{p_u}$ , we can deduce that $m$ can take at most $p_1^{r_{(1,y)}} p_2^{r_{(2,y)}} \cdots p_n^{r_{(n,y)}}$ distinct values (i.e., $y$ distinct values).
Based on the above discussion, for $i = ky + m$, it can take at most $ty$ distinct values modulo $d$ (where $ty < xy = d$). 

On the other hand, according to the Chinese Remainder Theorem

\begin{equation*}
    i=\sum_u\left( (i_u)_{0}p_u^{0}+(i_u)_{1}p_u^{1}+...+(i_u)_{m_u-1}p_u^{m_u-1}\right)M_uN_u.
\end{equation*}
By the uniqueness of the Chinese Remainder Theorem, $i$ modulo $d$ should take exactly $d$ distinct values. This contradicts our previous conclusion of having fewer than $d$ distinct values (where $ty < xy = d$). Therefore, it follows that $k$ modulo $x$ must take exactly $x$ distinct values.\par

An analogous process can be applied to prove that $m$ modulo $y$ will also take precisely $y$ distinct values. Similarly, the same conclusion holds for the subscript $l$ and $s$ in the subscript $j = lx + s$ of the quantum state $\ket{b_j}$ in Eq.(\ref{purekd}).\par
\end{proof}

%%==============================%%

\section{The proof of Lemma \ref{dlemma1}}\label{The proof of Lemma5}

\begin{proof}
We expand both sides of Eq.(\ref{dlemma1main}) on the basis $\ket{a_i}$. First, analyze the left-hand side
    \begin{equation}\label{sumskk}
    \begin{aligned}
        \psi_{ms}=&\frac{1}{x}\left(\sum_{k}\omega_{x}^{sk}\ket{a_{ky+m}}\right)\times\left(\sum_{k'}\omega_{x}^{-sk'}\bra{a_{k'y+m}}\right)\\
        =&\frac{1}{x}\sum_{k,k'}\omega_{x}^{s(k-k')}\ket{a_{ky+m}}\bra{a_{k'y+m}}.
    \end{aligned}
    \end{equation}
From Eq.(\ref{dindexi}) and Eq.(\ref{dindexj}), the exponent of $\omega_{x}$ can be written as
\begin{equation}\label{skexpension}
\begin{aligned}
    s(k-k')&=\left(\sum_v(\bm{j_v})_{0...r_{(v,x)}-1}M_vN_v\right)\times\left(\sum_u\Delta^{\bm{i_u}\bm{i_u'}}_{r_{(u,y)}...r_u-1}(x/p_u^{r_{(u,x)}})N_u\right).
\end{aligned} 
\end{equation}
Note that $M_v = d/p_v^{r_v}$. If $u \neq v$, then $p_u^{r_{(u,x)}}$ divides $M_v$. It follows 
\begin{equation*}
    M_v(x/p_u^{r_{(u,x)}})=(M_{v}/p_u^{r_{u,x}})x=0 \mod x.
\end{equation*}
Therefore, the exponent of $\omega_{x}$ can be simplified to
\begin{equation}\label{ske2}
\begin{aligned}   \sum_u&\left((\bm{j_u})_{0...r_{(u,x)}-1}M_uN_u\right)\left(\Delta^{\bm{i_u}\bm{i_u'}}_{r_{(u,y)}...r_u-1}(x/p_u^{r_{(u,x)}})N_u\right).
\end{aligned}
\end{equation}

The main different term of left-hand side and right-hand side of the Eq.(\ref{dlemma1main}) is about $p_1$. We divide Eq.(\ref{ske2}) into two parts: the case of $u=1$ and that of $u\neq 1$.

If $u\neq 1$, Eq.(\ref{ske2}) is reducible to 
\begin{equation}
    \sum_{u\neq1}(\bm{j_u})_{0...r_{(u,x)}-1}\Delta^{\bm{i_u}\bm{i_u'}}_{r_{(u,y)}...r_u-1}M_uN_u^2(x/p_u^{r_{(u,x)}}):=Sum(1).
\end{equation}.

If $u=1$, Eq.(\ref{ske2}) is reducible to
\begin{equation}\label{newindexdenote}
\begin{aligned}
&\left((\bm{j_1})_{0...r_{(1,x)}-1}M_1N_1\right)\left(\Delta^{\bm{i_1}\bm{i_1'}}_{r_{(1,y)}...r_1-1}(x/p_1^{r_{(1,x)}})N_1\right)\\
=&\left((\bm{j_1})_{0...r_{(1,x)-2}}+p_1^{r_{(1,x)}-1}(\bm{j_1})_{r_{(1,x)}-1}\right)\left(\Delta^{\bm{i_1}\bm{i_1'}}_{r_{(1,y)}}+p_1\Delta^{\bm{i_1}\bm{i_1'}}_{r_{(1,y)}+1...r_1-1}\right)M_1{N_1}^2(x/p_1^{r_{(1,x)}})\\
=&(\bm{j_1})_{0...r_{(1,x)-2}}\left(\Delta^{\bm{i_1}\bm{i_1'}}_{r_{(1,y)}}+p_1\Delta^{\bm{i_1}\bm{i_1'}}_{r_{(1,y)}+1...r_1-1}\right)M_1{N_1}^2(x/p_1^{r_{(1,x)}})\\
&+(\bm{j_1})_{r_{(1,x)}-1}\Delta^{\bm{i_1}\bm{i_1'}}_{r_{(1,y)}}M_1{N_1}^2(x/p_1) \mod x   \\
:=&Sum(2)+Sum(3) \mod x.
\end{aligned}
\end{equation}
Hence, the exponent of $\omega_{x}$ in Eq.(\ref{sumskk})
\begin{equation}\label{termofwx}
s(k-k')=Sum(1)+Sum(2)+Sum(3) \mod x.
\end{equation}
Notice that in the terms of Eq.(\ref{termofwx}), only $Sum(3)$ contains the index $j_{(1,r_{(1,x)}-1)}$. Therefore, in the following discussion, we will isolate this term for separate analysis.\par
Note that $k=\sum_u(\bm{i_u})_{r_{(u,y)...r_u-1}}(x/{p_u^{r_{(u,x)}}})N_u$ from Eq.(\ref{dindexi}). Denote $\{(i_{u'})_{t}\}=({(i_{u'})_{t}}p_{u'}^t)(x/{p_{u'}^{r_{(u',x)}}})N_{u'}$ for $t\geq r_{(u',y)}$. Obviously, $\{(i_{u'})_{t}\}$ is a certain term in the summation expression of $k$. 
%we introduce a new notation to denote the deletion of a term from the summation expression of $k$,
%\begin{equation*}
%\begin{split}
%    k&=\sum_u(\bm{i_u})_{r_{(u,y)...r_u-1}}(x/{p_u^{r_{(u,x)}}})N_u-((i_{u'})_{t}p_{u'}^t)(x/{p_{u'}^{r_{(u',x)}}})N_{u'}+((i_{u'})_{t}p_{u'}^t)(x/{p_{u'}^{r_{(u',x)}}})N_{u'}\\
%    &:=\sum_u(\bm{i_u})_{r_{(u,y)...r_u-1}}(x/{p_u^{r_{(u,x)}}})N_u-\{(i_{u'})_{t}\}+\{(i_{u'})_{t}\}
%\end{split}   
%\end{equation*}
%where $\{(i_{u'})_{t}\}=: ({(i_{u'})_{t}}p_{u'}^t)(x/{p_{u'}^{r_{(u',x)}}})N_{u'}$ for $t\geq r_{(u',y)}$.  
The meaning of $k-\{(i_{u'})_t\}$ is to remove the term $\{(i_{u'})_t\}$ from the summation expression of $k$. 
%Although this notation is identical to that in (\ref{determ}), it represents a different expression, but we still retain this notation in the present proof.\par
The subscript of sum can be written as
\begin{equation*}
    \sum_{k}=\sum_{(i_{u'})_{t}}\sum_{k-\{(i_{u'})_{t}\}}
\end{equation*}.\par

Applying Eq.(\ref{sumskk}), we proceed to perform the summation of $\psi_{ms}$ over $(j_{1})_{r_{(1,x)}-1}$, i.e.,
\begin{equation}\label{esumj}
\begin{aligned}
     \sum_{(j_1)_{r_{(1,x)}-1}}\psi_{ms}=&\frac{1}{x}\sum_{(j_1)_{r_{(1,x)}-1}}\sum_{\substack{k,k'}}\omega_{x}^{s(k-k')}\ket{a_{ky+m}}\bra{a_{k'y+m}}\\
     =&\frac{1}{x}\sum_{\substack{k,k'}}\sum_{(j_1)_{r_{(1,x)}-1}}\omega_{x}^{s(k-k')}\ket{a_{ky+m}}\bra{a_{k'y+m}}\\
     =&\frac{1}{x}\sum_{\substack{k,k'}}\omega_{x}^{Sum(1)}\omega_{x}^{Sum(2)}\left(\sum_{(j_1)_{r_{(1,x)}-1}}\omega_{x}^{Sum(3)}\right)\ket{a_{ky+m}}\bra{a_{k'y+m}}.
\end{aligned}
\end{equation}
Note that the sum notation can be rewritten as
\begin{equation*}
    \begin{split}
        \sum_{\substack{k,k'}}=\sum_{\substack{k-\{(i_1)_{r_{(1,y)}}\}\\k'-\{(i'_1)_{r_{(1,y)}}\}}}\left(\sum_{(i_1)_{r_{(1,y)}}=(i_1')_{r_{(1,y)}}}+\sum_{(i_1)_{r_{(1,y)}}\neq(i_1')_{r_{(1,y)}}}\right).
    \end{split}
\end{equation*}
For $Sum(3)$, consider whether $\Delta^{\bm{i_1}\bm{i_1'}}_{r_{(1,y)}}$ is equal to zero.

\text{Case} 1. $\Delta^{\bm{i_1}\bm{i_1'}}_{r_{(1,y)}}\neq 0$, i.e., $(i_1)_{r_{(1,y)}} - (i_{1}')_{r_{(1,y)}}\neq 0$.
Then
\begin{equation*}
    \begin{aligned}
        \omega_{x}^{Sum(3)}
        =&\omega_{x}^{(\bm{j_1})_{r_{(1,x)}-1}\Delta^{\bm{i_1}\bm{i_1'}}_{r_{(1,y)}}M_1{N_1}^2(x/p_1)}
        =\omega_{p_1}^{(\bm{j_1})_{r_{(1,x)}-1}\Delta^{\bm{i_1}\bm{i_1'}}_{r_{(1,y)}}M_1{N_1}^2}.
    \end{aligned}
\end{equation*}
Since the index $(j_1)_{r_{(1,x)}-1}\in \mathbb{Z}_{p_1}$, $\sum_{(j_1)_{r_{(1,x)}-1}}\omega_{x}^{Sum(3)}=0$.\par

By Case 1,  Eq.(\ref{esumj}) is equal to
\begin{equation}\label{newd7}
    \frac{1}{x}\sum_{\substack{k-\{(i_1)_{r_{(1,y)}}\}\\k'-\{(i'_1)_{r_{(1,y)}}\}}}\sum_{(i_1)_{r_{(1,y)}}=(i_1')_{r_{(1,y)}}}\omega_{x}^{Sum(1)}\omega_{x}^{Sum(2)}\left(\sum_{(j_1)_{r_{(1,x)}-1}}\omega_{x}^{Sum(3)}\right)\ket{a_{ky+m}}\bra{a_{k'y+m}}.
\end{equation}

\text{Case} 2. $\Delta^{\bm{i_1}\bm{i_1'}}_{r_{(1,y)}}= 0$, i.e., $(i_1)_{r_{(1,y)}} - (i_{1}')_{r_{(1,y)}}=0$.
It follows
$\sum_{(j_1)_{r_{(1,x)}-1}}\omega_{x}^{Sum(3)}=p_1$. 
Therefore, Eq.(\ref{newd7}) can be written as
\begin{equation}\label{deqleft}
    \begin{aligned}
    &\sum_{(j_1)_{r_{(1,x)}-1}}\psi_{ms}
        =\frac{p_1}{x}\sum_{(i_1)_{r_{(1,y)}}=(i_1')_{r_{(1,y)}}}\sum_{\substack{k-\{(i_1)_{r_{(1,y)}}\}\\k'-\{(i_1')_{r_{(1,y)}}\}}}\omega_{x}^{Sum(1)}\omega_{x}^{Sum(2)}\ket{a_{ky+m}}\bra{a_{k'y+m}}.
    \end{aligned}
\end{equation}

According to the definition $x=\widetilde{x}p_1$,
\begin{equation*}
\begin{split}
    &r_{(1,\widetilde{x})}=r_{(1,x)}-1, \quad r_{(1,\widetilde{y})}=r_{(1,y)}+1,\\
    &r_{(u,\widetilde{x})}=r_{(u,x)},  \quad r_{(u,\widetilde{y})}=r_{(u,y)}, \quad \text{for} \quad u\neq1.
\end{split}
\end{equation*}
Then we obtain
\begin{equation}\label{leftsum1}
\begin{split}
    \omega_{x}^{Sum(1)}=&\omega_{x}^{\sum_{u\neq1}(\bm{j_u})_{0...r_{(u,x)}-1}\Delta^{\bm{i_u}\bm{i_u'}}_{r_{(u,y)}...r_u-1}M_uN_u^2(x/p_u^{r_{(u,x)}})}\\
    =&\omega_{x/p_1}^{\sum_{u\neq1}(\bm{j_u})_{0...r_{(u,x)}-1}\Delta^{\bm{i_u}\bm{i_u'}}_{r_{(u,y)}...r_u-1}M_uN_u^2(x/p_u^{r_{(u,x)}})/p_1}\\
    =&\omega_{\widetilde{x}}^{\sum_{u\neq1}(\bm{j_u})_{0...r_{(u,\widetilde{x})}-1}\Delta^{\bm{i_u}\bm{i_u'}}_{r_{(u,\widetilde{y})}...r_u-1}M_uN_u^2(\widetilde{x}/p_u^{r_{(u,\widetilde{x})}})}.
\end{split}
\end{equation}
Notice that $(i_1)_{r_{(1,y)}}=(i_1')_{r_{(1,y)}}$, i.e., $\Delta^{\bm{i_1}\bm{i_1'}}_{r_{(1,y)}}=0$. Hence
\begin{equation}\label{leftsum2}
    \begin{split}
    \omega_{x}^{Sum(2)}=&\omega_{x}^{(\bm{j_1})_{0...r_{(1,x)-2}}\left(\Delta^{\bm{i_1}\bm{i_1'}}_{r_{(1,y)}}+p_1\Delta^{\bm{i_1}\bm{i_1'}}_{r_{(1,y)}+1...r_1-1}\right)M_1{N_1}^2(x/p_1^{r_{(1,x)}})}\\
        %=&\omega_{\widetilde{x}}^{(\bm{j_1})_{0...r_{(1,\widetilde{x})-1}}\Delta^{\bm{i_1}\bm{i_1'}}_{r_{(1,\widetilde{y})}...r_1-1}M_1{N_1}^2(\widetilde{x}/p_1^{r_{(1,x)}})p_1}\\
        =&\omega_{\widetilde{x}}^{(\bm{j_1})_{0...r_{(1,\widetilde{x})-1}}\Delta^{\bm{i_1}\bm{i_1'}}_{r_{(1,\widetilde{y})}...r_1-1}M_1{N_1}^2(\widetilde{x}/p_1^{r_{(1,\widetilde{x})}})}
    \end{split}
\end{equation}

We now consider the right-hand side of Eq.(\ref{dlemma1main}) in Lemma \ref{dlemma1}.

\begin{equation}
\begin{aligned}
    \psi_{\widetilde{m}\widetilde{s}}
    =&\frac{1}{\widetilde{x}}\sum_{\widetilde{k},\widetilde{k}'}\omega_{\widetilde{x}}^{\widetilde{s}(\widetilde{k}-\widetilde{k}')}\ket{a_{\widetilde{k}\widetilde{y}+\widetilde{m}}}\bra{a_{\widetilde{k}'\widetilde{y}+\widetilde{m}}}.
\end{aligned}   
\end{equation}

Similarly, we categorize the summation terms in the exponent of $\omega_{{\widetilde{x}}}$, then we sum $\psi_{\widetilde{m}\widetilde{s}}$ over $(\widetilde{i_1})_{r_{(1,y)}}$. Finally, we obtain 
\begin{equation}\label{deqr}
    \sum_{(\widetilde{i_1})_{r_{(1,y)}}}\psi_{\widetilde{m}\widetilde{s}}=\frac{1}{{\widetilde{x}}}
    \sum_{(\widetilde{i}_1)_{r_{(1,y)}}}
    \sum_{\widetilde{k},\widetilde{k}'}
    \omega_{{\widetilde{x}}}^{\widetilde{Sum}(1)}\omega_{{\widetilde{x}}}^{\widetilde{Sum}(2)}\ket{a_{\widetilde{k}{\widetilde{y}}+\widetilde{m}}}\bra{a_{\widetilde{k}'{\widetilde{y}}+\widetilde{m}}},
\end{equation}
where
\begin{equation}\label{rightsum1}
\begin{aligned}
     \widetilde{Sum}(1)&:=\sum_{u\neq1}(\bm{\widetilde{j}_{u}})_{0...r_{(u,\widetilde{x})}-1}\Delta^{\bm{\widetilde{i_u}}\bm{\widetilde{i_u'}}}_{r_{(u,\widetilde{y})}...r_u-1}M_uN_u^2({\widetilde{x}}/p_u^{r_{(u,\widetilde{x})}})\\
\end{aligned},
\end{equation}
\begin{equation}\label{rightsum2}
    \begin{aligned}
        \widetilde{Sum}(2)&:=(\bm{\widetilde{j}_1})_{0...r_{(1,\widetilde{x})-1}}\Delta^{\bm{\widetilde{i_1}}\bm{\widetilde{i_1'}}}_{r_{(1,\widetilde{y})}...r_1-1}M_1N_1^2({\widetilde{x}}/p_1^{r_{(1,\widetilde{x})}})
    \end{aligned}.
\end{equation}

%We denote $\mathbb{I}_{\widetilde{k},\widetilde{k}'}$ as

%\begin{equation}
%\{\widetilde{i}_{(1,m_{(1,y)}+1)}...\widetilde{i}_{(1,m_1-1)}\}\bigcup\left\{\bigcup_{u\neq1}\{\widetilde{i}_{(u,m_{(u,{y})})}'...\widetilde{i}_{(u,m_u-1)}'\}\right\}.
%\end{equation}

Now we check the subscript of $\ket{a}$ and $\bra{a}$ in Eq.(\ref{deqleft}) and Eq.(\ref{deqr}). Recall that with Eq.(\ref{epofiwithd}), the subscript of $\ket{a}$ have these expressions:
\begin{equation}\label{subcompare}
    \begin{split}
        &ky=\sum_{u}p_{u}^{r_{u,y}}(\bm{i_u})_{r_{(u,y)}...r_u-1}M_uN_u; \quad \quad \widetilde{k}\widetilde{y}=\sum_{u}p_{u}^{r_{(u,\widetilde{y})}}(\bm{\widetilde{i_u}})_{r_{(u,\widetilde{y})}...r_u-1}M_uN_u;\\
        &m=\sum_{u}(\bm{i_u})_{0...r_{(u,y)}-1}M_uN_u;\quad \quad\quad\; \quad \widetilde{m}=\sum_{u}(\bm{\widetilde{i_u}})_{0...r_{(u,\widetilde{y})}-1}M_uN_u.
    \end{split}
\end{equation}
Notice the form of subscript of the sum notation in Eq.(\ref{deqleft}),we rewrite the subscript of $\ket{a}$ as
\begin{equation*}
    ky+m=\left(k-\{(i_1)_{r_{(1,y)}}\}\right)y+\left(\{(i_1)_{r_{(1,y)}}\}y+m\right).
\end{equation*}
Consider the 1st term and recall that $r_{(1,\widetilde{y})}=r_{(1,y)}+1$ and $r_{(u,\widetilde{y})}=r_{(u,y)}$ for $u\neq1$:
\begin{equation*}
    \begin{split}
        \left(k-\{(i_1)_{r_{(1,y)}}\}\right)y=&\sum_{u\neq1}p_{u}^{r_{(u,y)}}(\bm{i_u})_{r_{(u,y)}...r_u-1}M_uN_u+p_{1}^{r_{(1,y)}+1}(i_1)_{r_{(1,y)}+1...r_1-1}M_1N_1\\
        =&\sum_{u}p_{u}^{r_{(u,\widetilde{y})}}(\bm{i_u})_{r_{(u,\widetilde{y})}...r_u-1}M_uN_u.
    \end{split}
\end{equation*}
It has the same form of $\widetilde{k}\widetilde{y}$ in Eq.(\ref{subcompare}). Similar for $\left(\{(i_1)_{r_{(1,y)}}\}y+m\right)$:
\begin{equation*}
    \left(\{(i_1)_{r_{(1,y)}}\}y+m\right)=\sum_{u}(\bm{i_u})_{0...r_{(u,\widetilde{y})}}M_uN_u.
\end{equation*}
It has the same form of $\widetilde{m}$ in Eq.(\ref{subcompare}). The subscript of $\bra{a}$ is similar.\par
%%差描述
When the subscripts of $\ket{a}$ and $\bra{a}$ in Eq.(\ref{deqleft}) are equal to those in Eq.(\ref{deqr}), respectively, i.e., $ky+m=\widetilde{k}\widetilde{y}+\widetilde{m}\; , \;k'y+m=\widetilde{k'}\widetilde{y}+\widetilde{m}$, which means 
\begin{equation*}
(i_u)_{t}=(\widetilde{i_u})_{t}\;,\;(i_u')_{t}=(\widetilde{i_u'})_{t};\quad \quad t\in[0,m_u-1]    
\end{equation*}
for $n\geq u\geq1$. With the expression Eq.(\ref{leftsum1}), Eq.(\ref{rightsum1}), Eq.(\ref{leftsum2}) and Eq.(\ref{rightsum2}), when $ky+m=\widetilde{k}\widetilde{y}+\widetilde{m}\; , \;k'y+m=\widetilde{k'}\widetilde{y}+\widetilde{m}$, we have $\omega_{x}^{Sum(1)}=\omega_{\widetilde{x}}^{\widetilde{Sum(1)}}$ and $\omega_{x}^{Sum(2)}=\omega_{\widetilde{x}}^{\widetilde{Sum(2)}}$. That means that Eq.(\ref{deqleft}) is equal to Eq.(\ref{deqr}). This completes the proof of Lemma \ref{dlemma1}.\par
\end{proof}

%%==============================================%%

\section{The proof of Theorem \ref{dthm}}\label{proofofthm2}

\begin{proof}
Suppose $d=p_1^{r_1}...p_n^{r_n}$ and $x_0=p_1^{r_1}...p_{\alpha}^{r_{\alpha}}$, $y_0=p_{\alpha+1}^{r_{\alpha+1}}...p_n^{r_n}$. Let $L=\sum_{i=1}^{n}r_i$. From Example \ref{example}, we can easily obtain a fact that the length of any path from  $v_{x_0}$ to $v_{y_0}$ in graph $G({x_0})$ is $L$. Without loss of generality, we select a representative path
\begin{equation}\label{xofpath}
    v_{path}:v_{x_{0}}\rightarrow v_{x_{1}}\rightarrow...\rightarrow v_{x_{L}},
\end{equation} 

where
\begin{equation*}
  \begin{aligned} 
   x_t=\left \{
   \begin{array}{ll}
    \frac{x_0}{p_1^{\delta_t}},     & \quad \delta_t = t \; \text{if}  \; 0 \leq t \leq r_1,\\[10pt]
    \frac{x_0}{p_{h}^{\delta_t}\prod_{i=1}^{h-1}p_i^{r_i}},    &\quad \delta_t = t-\sum_{i=1}^{h-1}r_i \; \text{if}  \;  \sum_{i=1}^{h-1}r_i < t \leq \sum_{i=1}^{h}r_i \; \text{and} \; 2<h \leq \alpha, \\[10pt]
    p_{\alpha+1}^{\delta_t},     &\quad \delta_t = t-\sum_{i=1}^{\alpha}r_i \; \text{if}  \; \sum_{i=1}^{\alpha}r_i < t \leq \sum_{i=1}^{\alpha+1}r_i,\\[10pt]
    p_{h}^{\delta_t}\prod_{i=\alpha+1}^{h-1}p_i^{r_i},  &\quad \delta_t = t-\sum_{i=1}^{h-1}r_i \; \text{if}  \;    \sum_{i=1}^{h-1}r_i < t \leq \sum_{i=1}^{h}r_i \; \text{and} \; \alpha+2<h \leq n.\\
   \end{array}
   \right.   
  \end{aligned} 
\end{equation*}
For Example \ref{example}, we apply this method to construct the path, which gives path Eq.(\ref{path-example}). The method of proof for different path is similar. Here we prove the result for the above path Eq.(\ref{xofpath}).

The inclusion $\mathrm{conv}(v_{path})\subseteq \mathrm{KD}_{\mathcal{A,B}}^+\bigcap \mathrm{span}_{\mathbb{R}}(v_{path})$ follows analogously to the proof technique of Lemma \ref{lemma2}. Here we only consider the reverse inclusion. The proof proceeds in two steps.        

\textbf{Step 1}. Suppose a state $\rho\in \mathrm{KD}_{\mathcal{A,B}}^+\bigcap \mathrm{span}_{\mathbb{R}}(v_{path})$. It follows $\rho\in\mathrm{span}_{\mathbb{R}}(v_{path})$. Then $\rho$ can be expressed as a real linear combination of KD-classical pure states, i.e.,
    \begin{equation}\label{dspanrho}
\rho=\sum_{m_0,s_0}\lambda_{m_0,s_0}\psi_{m_0,s_0}+\sum_{m_1,s_1}\lambda_{m_1,s_1}\psi_{m_1,s_1}+...+\sum_{m_{L},s_{L}}\lambda_{m_{L},s_{L}}\psi_{m_{L},s_{L}}.
    \end{equation}
Note that  the $m_t$ and $s_t$ with superscripts correspond to the prime factorization of $d=x_ty_t$, and to the representation of KD-classical pure states on vertex $v_{x_t}$.\par

Recall that $(\bm{i_u})_{0...r_{(u,y)}-1}=0$ if $r_{(u,y)}=0$, and $(\bm{i_u})_{r_{(u,y)}...r_u-1}=0$ if $r_{(u,y)}=r_u$ in Eq.(\ref{dindexi}). Since $y_0=p_{\alpha+1}^{r_{\alpha+1}}...p_n^{r_n}$ and $x_0=p_{1}^{r_1}...p_{\alpha}^{r_{\alpha}}$, notice that $r_{(u,y_0)}=r_u$ for $\alpha+1\leq u\leq n$ and $r_{(u,y_0)}=0$ for $1\leq u\leq \alpha$, then $m_0$ and $s_0$ can be written as
\begin{equation*}
\begin{split}
    m_0= \sum_{u=\alpha+1}^{n}(\bm{i_{u}})_{0...r_u-1}M_uN_u, \quad
    s_0=\sum_{v=1}^{\alpha}(\bm{j_v})_{0...r_v-1}M_vN_v
\end{split}
\end{equation*}

Employing Eq.(\ref{m-m'}) to analyze the 1st and 2nd summation terms in Eq.(\ref{dspanrho}), we have
\begin{equation}\label{determ}
\begin{split}
    &m_0=m_1-({i_1})_0M_1N_1 := m_1-\{({i_1})_{0}\}\\
    &s_1=s_0-({j_1})_{r_1-1}M_1N_1p_1^{r_1-1} := s_0-\{({j_1})_{r_1-1}\}.
\end{split}
\end{equation}
Here we use a new notation $\{({i_u})_{t}\}=({i_u})_{t}M_uN_up_u^t$ to simplify the original notation. The subscript of the summation notation can be rewritten as
\begin{equation*}
\begin{split}
    \sum_{m_1}=\sum_{m_1-\{({i_1})_0\}}\sum_{({i_1})_0}=\sum_{m_0}\sum_{({i_1})_0}; \quad \sum_{s_0}=\sum_{s_0-\{({j_1})_{r_1-1}\}}\sum_{({j_1})_{r_1-1}}=\sum_{s_1}\sum_{({j_1})_{r_1-1}}
\end{split}   
\end{equation*}

%We similarly construct a new set of coefficients $\lambda''$ subject to the non-negativity constraint $\lambda''\geq0$.\par
%First, we analyze the first and second summation terms in (\ref{dspanrho}). Specifically, the summation indices $t^{(1)}\in\mathbb{I}^{1}$ correspond to

%\begin{equation*}
%%%%%    \end{equation*}
Let's consider the coefficients of KD-classical pure state corresponding to $v_{x_0}$ and $v_{x_1}$.

    \begin{equation*}
        \begin{aligned}
            &\sum_{m_0,s_0}\lambda_{m_0,s_0}\psi_{m_0,s_0}+\sum_{m_1,s_1}\lambda_{m_1,s_1}\psi_{m_1,s_1}\\
            =&\sum_{\substack{m_0\\s_1-\{({j_1})_{r_1-1}\}}}\sum_{({j_1})_{r_1-1}}\lambda_{m_0,s_0}\psi_{m_0,s_0}+\sum_{\substack{m_1-\{({i_1})_{0}\}\\s_1}}\sum_{({i_1})_{0}}\lambda_{m_1,s_1}\psi_{m_1,s_1}\\
            =&\sum_{\substack{m_0\\s_1-\{({j_1})_{r_1-1}\}}}\sum_{({j_1})_{r_1-1}}(\lambda_{m_0,s_0}-\lambda_{m_0,s_1}^{\mathrm{min}})\psi_{m_0,s_0}+\sum_{\substack{m_1-\{({i_1})_{0}\}\\s_1}}\sum_{({i_1})_{0}}(\lambda_{m_1,s_1}+\lambda_{m_0,s_1}^{\mathrm{min}})\psi_{m_1,s_1}\\
            =&\sum_{\substack{m_0,s_0}}(\lambda_{m_0,s_0}-\lambda_{m_0,s_1}^{\mathrm{min}})\psi_{m_0,s_0}+\sum_{\substack{m_1,s_1}}(\lambda_{m_1,s_1}+\lambda_{m_0,s_1}^{\mathrm{min}})\psi_{m_1,s_1}\\
            :=&\sum_{\substack{m_0,s_0}}(\lambda_{m_0,s_0}-\lambda_{m_0,s_1}^{\mathrm{min}})\psi_{m_0,s_0}+\sum_{\substack{m_1,s_1}}\widetilde\lambda_{m_1,s_1}\psi_{m_1,s_1}
        \end{aligned}.
    \end{equation*} 
The 2nd equality holds since $\sum_{({j_1})_{r_1-1}}\psi_{m_0,s_0}=\sum_{({i_1})_0}\psi_{m_1,s_1}$
by Lemma \ref{dlemma1}, where
\begin{equation*}
    \lambda_{m_0,s_1}^{\mathrm{min}}=\min_{({j_1})_{r_1-1}}\{\lambda_{m_0,s_0}\}.
\end{equation*}
For simplicity, we denote $(\lambda_{m_1,s_1}+\lambda_{m_0,s_1}^{\mathrm{min}})=\widetilde\lambda_{m_1,s_1}$ in the last equality.
Now, the coefficients $\lambda_{m_0,s_0}-\lambda_{m_0,s_1}^{\mathrm{min}}\geq 0$. Next we consider the coefficients of KD-classical pure state corresponding to $v_{x_{t-1}}$ and $v_{x_t}$ with $1< t\leq L$.

%$\lambda_{m_0,s_0}-\lambda_{m_0,s_1}^{\mathrm{min}}\geq 0$.
%Denoted $(\lambda_{m_1,s_1}+\lambda_{m_0,s_1}^{\mathrm{min}})=\widetilde\lambda_{m_1,s_1}$. 

Case 1. Set  $\sum_{i=1}^{0}r_i=0$. Then that $\sum_{i=1}^{h-1}r_i <t \leq \sum_{i=1}^{h}r_i$ with $1< h \leq \alpha$ contains two cases for $x_t$, i.e.,  $0<t\leq r_1$ or $\sum_{i=1}^{h-1}r_i <t \leq \sum_{i=1}^{h}r_i$ and $2< h \leq \alpha $.

%For $x_t$ with $\sum_{i=1}^{h-1}r_i <t \leq \sum_{i=1}^{h}r_i\leq \sum_{i=1}^{\alpha}r_i$ and $2<t$. Denoted $\sum_{i=1}^{h-1}r_i=0$ when $h=1$ for $\sum_{i=1}^{0}r_i=0<t\leq \sum_{i=1}^{1}r_i$. 
By Eq.(\ref{m-m'}), we have
\begin{equation*}
    \begin{split}
        m_t=m_{t-1}+\{({i_h})_{\delta_t-1}\},\quad
        s_t=s_{t-1}-\{({j_h})_{r_h-\delta_t}\}.
    \end{split}
\end{equation*}
where $\delta_t=t-\sum_{i=1}^{h-1}r_i$.
Following the same procedure as for $v_{x_0}$ and $v_{x_1}$ above, consider the coefficient of the $(t-1)$-th and $t$-th terms as 
\begin{equation*}
\begin{split}
    ...+&\sum_{m_{t-1},s_{t-1}}\widetilde{\lambda}_{m_{t-1},s_{t-1}}\psi_{m_{t-1},s_{t-1}}+\sum_{m_t,s_t}\lambda_{m_t,s_t}\psi_{m_t,s_t}+...\\
    =...+&\sum_{\substack{m_{t-1}\\s_{t-1}-\{({j_h})_{r_1-\delta_t}\}}}\sum_{({j_h})_{r_1-\delta_t}}(\widetilde{\lambda}_{m_{t-1},s_{t-1}}-\widetilde{\lambda}^{\min}_{m_{t-1},s_t})\psi_{m_{t-1},s_{t-1}}\\
    &+\sum_{\substack{m_t-\{({i_h})_{\delta_t-1}\}\\s_t}}\sum_{({i_h})_{\delta_t-1}}(\lambda_{m_t,s_t}+\widetilde{\lambda}^{\min}_{m_{t-1},s_t})\psi_{m_t,s_t}+...\\
    =...+&\sum_{\substack{m_{t-1},s_{t-1}}}(\widetilde{\lambda}_{m_{t-1},s_{t-1}}-\widetilde{\lambda}^{\min}_{m_{t-1},s_t})\psi_{m_{t-1},s_{t-1}}+\sum_{\substack{m_t,s_t}}(\lambda_{m_t,s_t}+\widetilde{\lambda}^{\min}_{m_{t-1},s_t})\psi_{m_t,s_t}+...\\
    :=...+&\sum_{\substack{m_{t-1},s_{t-1}}}(\widetilde{\lambda}_{m_{t-1},s_{t-1}}-\widetilde{\lambda}^{\min}_{m_{t-1},s_t})\psi_{m_{t-1},s_{t-1}}+\sum_{\substack{m_t,s_t}}\widetilde\lambda_{m_t,s_t}\psi_{m_t,s_t}+...
\end{split} 
\end{equation*}
%\begin{equation*}
%    \widetilde{\lambda}_{m_{t-1},s_{t-1}}-\widetilde{\lambda}^{\min}_{m_{t-1},s_t},
%\end{equation*}
where $\widetilde{\lambda}_{m_{t-1},s_t}^{\min}=\min_{({j_h})_{r_h-\delta t}}\{\widetilde{\lambda}_{m_{t-1},s_{t}+\{({j_h})_{r_h-\delta_t}\}}\}$. 
The first equality follows form $\sum_{({j_1})_{r_1-1}}\psi_{m_{t-1},s_{t-1}}=\sum_{({i_1})_0}\psi_{m_t,s_t}$
by Lemma \ref{dlemma1}. Now the coefficient of the $t$-th term, $\widetilde{\lambda}_{m_{t-1},s_{t-1}}-\widetilde{\lambda}^{\min}_{m_{t-1},s_t}$, is nonnegative.

Case 2. For $\sum_{i=1}^{h-1}r_i<t \leq \sum_{i=1}^{h}r_i$ with $\alpha+1 \leq h \leq n$, a similar discussion to that in Case 1 leads to
\begin{equation*}
    \begin{split}
        m_t=m_{t-1}-\{(i_h)_{r_h-\delta_t}\},\quad
        s_t=s_{t-1}+\{(j_h)_{\delta_t-1}\}.
    \end{split}
\end{equation*}
Thus
\begin{equation*}
\begin{split}
    ...+&\sum_{m_{t-1},s_{t-1}}\widetilde{\lambda}_{m_{t-1},s_{t-1}}\psi_{m_{t-1},s_{t-1}}+\sum_{m_t,s_t}\lambda_{m_t,s_t}\psi_{m_t,s_t}+...\\
    =...+&\sum_{\substack{m_{t-1}-\{({i_h})_{r_h-\delta_t}\}\\s_{t-1}}}\sum_{({i_h})_{r_h-\delta_t}}(\widetilde{\lambda}_{m_{t-1},s_{t-1}}-\widetilde{\lambda}^{\min}_{m_{t},s_{t-1}})\psi_{m_{t-1},s_{t-1}}\\
    &+\sum_{\substack{m_t\\s_t-\{({j_h})_{\delta_t-1}\}}}\sum_{({j_h})_{\delta_t-1}}(\lambda_{m_t,s_t}+\widetilde{\lambda}^{\min}_{m_{t},s_{t-1}})\psi_{m_t,s_t}+...\\
    =...+&\sum_{\substack{m_{t-1},s_{t-1}}}(\widetilde{\lambda}_{m_{t-1},s_{t-1}}-\widetilde{\lambda}^{\min}_{m_{t},s_{t-1}})\psi_{m_{t-1},s_{t-1}}
    +\sum_{\substack{m_t,s_t}}(\lambda_{m_t,s_t}+\widetilde{\lambda}^{\min}_{m_{t},s_{t-1}})\psi_{m_t,s_t}+...\\
    :=...+&\sum_{\substack{m_{t-1},s_{t-1}}}(\widetilde{\lambda}_{m_{t-1},s_{t-1}}-\widetilde{\lambda}^{\min}_{m_{t},s_{t-1}})\psi_{m_{t-1},s_{t-1}}
    +\sum_{\substack{m_t,s_t}}\widetilde{\lambda}_{m_t,s_t}\psi_{m_t,s_t}+...
\end{split} 
\end{equation*}
where $\widetilde{\lambda}_{m_t,s_{t-1}}^{\min}=\min_{(i_h)_{r_h-\delta t}}\{\widetilde{\lambda}_{m_{t}+\{(i_h)_{r_h-\delta_t}\},s_{t-1}}\}$. % 把先把写清楚与变化下标关联起来
Note that the index of the minimum here is different from that in Case 1. The coefficient of the $(t-1)$-th term, $\widetilde{\lambda}_{m_{t-1},s_{t-1}}-\widetilde{\lambda}^{\min}_{m_t,s_{t-1}}$, has become nonnegative.

Now the new expression of $\rho$ is
\begin{equation*}
    \begin{split}
        \rho=&\sum_{m_0,s_0}(\lambda_{m_0,s_0}-\lambda^{\min}_{m_0,s_1})\psi_{m_0,s_1}+\sum_{t=2}^{L_{\alpha}}\sum_{m_{t-1},s_t-1}(\widetilde{\lambda}_{m_{t-1},s_{t-1}}-\widetilde{\lambda}^{\min}_{m_{t-1},s_t})\psi_{m_{t-1},s_{t-1}}\\
        &+\sum_{t=L_{\alpha}+1}^{L}\sum_{m_{t-1},s_t-1}(\widetilde{\lambda}_{m_{t-1},s_{t-1}}-\widetilde{\lambda}^{\min}_{m_{t},s_{t-1}})\psi_{m_{t-1},s_{t-1}}+\sum_{m_L,s_L}\widetilde{\lambda}_{m_L,s_L}\psi_{m_L,s_L},
    \end{split}
\end{equation*}
 where $L_{\alpha}=\sum_{i=1}^{\alpha}r_i$. All coefficients except $\widetilde{\lambda}_{m_L,s_L}$ are nonnegative.

\textbf{Step 2}. We prove that 
 the coefficients of the last term $\widetilde{\lambda}_{m_L,s_L}=\lambda_{m_{L},s_L}+\widetilde\lambda_{m_L,s_{L-1}}^{\mathrm{min}}$ are nonnegative. For the last vertex $v_{x_L}$, $x_L=p_{\alpha+1}^{r_{\alpha+1}}...p_n^{r_n}$. We have
    \begin{equation*}
        \begin{split}
            m_L=\sum_{u=1}^{\alpha}(\bm{i_u}
            )_{0...r_u-1}M_uN_u,\quad
            s_L=\sum_{v=\alpha+1}^{n}(\bm{j_v}
            )_{0...r_v-1}M_vN_v
        \end{split},
    \end{equation*}
and $m_L \in \mathbb{Z}_{y_L}$ and $s_L \in \mathbb{Z}_{x_L}$.
For any $m^*_L \in \mathbb{Z}_{y_L}, s^*_L \in \mathbb{Z}_{x_L}$,
we will expand $\widetilde{\lambda}_{m^*_L,s^*_L}$ as the sum of the original coefficient $\lambda$ and prove that the sum is nonnegative.\par

By definition we have $\widetilde{\lambda}_{m^*_L,s^*_L}=\lambda_{m^*_{L},s^*_L}+\widetilde\lambda_{m^*_L,s_{L-1}^*}^{\mathrm{min}}$, where $s_{L-1}^*=s_L^*-\{(j_n)_{r_n-1}\}$. Since $\widetilde\lambda_{m_L,s_{L-1}}^{\mathrm{min}} = \min_{(i_n)_{0}}\{\widetilde{\lambda}_{\substack{m_L+(i_n)_{0}\\s_{L-1}}}\}$,
there exists $(i_n^*)_{0}\in\mathbb{Z}_{p_n}$ such that
%    \begin{equation*}
%        \widetilde\lambda_{m_L,s_{L-1}}^{\mathrm{min}}=\widetilde{\lambda}_{\substack{m_L+(i^*_n)_{0}\\s_L-\{(j_n)_{r_n-1}\}}}:=\widetilde{\lambda}_{\substack{m_{L-1},s_{L-1}}}.
%    \end{equation*}
    \begin{equation*}
        \widetilde\lambda_{m^*_L,s_{L-1}^*}^{\mathrm{min}}=\widetilde{\lambda}_{\substack{m^*_L+(i^*_n)_{0}\\s_{L-1}^*}}:=\widetilde{\lambda}_{\substack{m^*_{L-1},s_{L-1}^*}}.
    \end{equation*}
Here we denote the subscript $m^*_L+(i^*_n)_{0}$ as $m_{L-1}^*$. 
%Here $\widetilde{\lambda}_{\substack{m^*_{L-1}, s_{L-1}^*}}$ is a certain element of the set $\{ \widetilde{\lambda}_{m_{L-1},s_{L-1}}\}$ .
Then $\widetilde{\lambda}_{m^*_L,s^*_L}=\lambda_{m^*_{L},s^*_L}+\widetilde{\lambda}_{\substack{m^*_{L-1},s_{L-1}^*}}$. Since $\widetilde{\lambda}_{m_{L-1},s_{L-1}}=\lambda_{m_{L-1},s_{L-1}} + \widetilde\lambda_{m_{L-1},s_{L-2}}^{\mathrm{min}}$,
    \begin{equation*}    \widetilde{\lambda}_{m^*_L,s^*_L}=\lambda_{m^*_{L},s^*_L}+\widetilde{\lambda}_{\substack{m^*_{L-1},s_{L-1}^*}}=\lambda_{m^*_{L},s^*_L}+\lambda_{\substack{m^*_{L-1},s_{L-1}^*}}+\widetilde\lambda_{m^*_{L-1},s_{L-2}^*}^{\mathrm{min}}.
    \end{equation*}
with $s_{L-2}^*=s_{L-1}^*-\{(j_n)_{r_n-2}\}$. 
During the step-by-step process of expressing $\widetilde\lambda_{m_t,s_t}$ and $\widetilde{\lambda}^{\min}_{m_{t-1},s_t}$ for $t>L_{\alpha}=\sum_{i=1}^{\alpha}r_i$, we gradually determine
\begin{equation*}
    \big\{(i_n^*)_{0}...(i_n^*)_{r_n-1}\big\}...\big\{(i^*_{\alpha+1})_{0}...(i^*_{\alpha+1})_{r_{\alpha+1}-1}\big\}.
\end{equation*}
and $\widetilde{\lambda}_{m^*_L,s^*_L}=\lambda_{m^*_L,s^*_L}+\lambda_{m^*_{L-1},s^*_{L-1}}+...+\widetilde\lambda_{m^*_{L_{\alpha}},s^*_{L_{\alpha}}}.$\par
For $t=L_{\alpha}$, $\widetilde{\lambda}_{m^*_{L_{\alpha}},s^*_{L_{\alpha}}}={\lambda}_{m^*_{L_{\alpha}},s^*_{L_{\alpha}}}+\widetilde\lambda^{\min}_{m^*_{L_{\alpha}-1},s^*_{L_{\alpha}}}$ where $m^*_{L_{\alpha}-1}=m^*_{L_{\alpha}}-\{({i_{\alpha}})_{r_{\alpha}-1}\}$.
Since $\widetilde\lambda^{\min}_{m_{L_{\alpha}-1},s_{L_{\alpha}}}=\min_{({j_{\alpha}})_{0}}\{\widetilde\lambda_{m_{L_{\alpha}-1},s_{L_{\alpha}}+({j}_{\alpha})_{0}}\}$, there exists $({j}^*_{\alpha})_{0} \in \mathbb{Z}_{p_{\alpha}} $ such that
\begin{equation*}
    \widetilde\lambda^{\min}_{m^*_{L_{\alpha}-1},s^*_{L_{\alpha}}}=\widetilde\lambda_{m^*_{L_{\alpha}-1},s^*_{L_{\alpha}}+({j}^*_{\alpha})_{0}}:=\widetilde\lambda_{m^*_{L_{\alpha}-1},s^*_{L_{\alpha}-1}}
\end{equation*}
Here we denote $s^*_{L_{\alpha}-1}=s^*_{L_{\alpha}}+({j}^*_{\alpha})_{0}$. So we have
\begin{equation*}
    \widetilde{\lambda}_{m^*_{L_{\alpha}},s^*_{L_{\alpha}}}={\lambda}_{m^*_{L_{\alpha}},s^*_{L_{\alpha}}}+{\lambda}_{m^*_{L_{\alpha}-1},s^*_{L_{\alpha}-1}}+\widetilde\lambda^{\min}_{m^*_{L_{\alpha}-2},s^*_{L_{\alpha}-1}}
\end{equation*}
with $m^*_{L_{\alpha}-2}=m^*_{L_{\alpha}-1}-\{({i_{\alpha}})_{r_{\alpha}-2}\}$.
During the similar process we gradually determine 
\begin{equation*}
    \big\{(j^*_{\alpha})_{0}...(j^*_{\alpha})_{r_{\alpha}-1}\big\}...\big\{(j_1^*)_{0}...(j_1^*)_{r_1-1}\big\},
\end{equation*}
and $\widetilde{\lambda}_{m^*_L,s^*_L}$ is written as the sum of the original coefficient $\lambda$, i.e.,
    \begin{equation}\label{sumofoc}
        \widetilde{\lambda}_{m^*_L,s^*_L}=\lambda_{m^*_L,s^*_L}+\lambda_{m^*_{L-1},s^*_{L-1}}+...+\lambda_{m^*_0,s^*_0}.
    \end{equation} 
    
Next, we construct $i'$ and $j'$ to show that Eq.(\ref{sumofoc}) is nonnegative. Let $i'$ and $j'$ be
%\begin{equation}\label{dindexij}
%    \begin{aligned}
%        i'=\sum_{u=1}^{\alpha}(\bm{i_u})_{0...r_u-1}M_uN_u+\sum_{u=\alpha+1}^{n}(\bm{i_u^*})_{0...r_u-1}M_uN_u=\sum_{u=1}^{n}(\bm{i'_u})_{0...r_u-1}M_uN_u\\
%        j'=\sum_{v=1}^{\alpha}(\bm{j_v^*})_{0...r_v-1}M_vN_v+\sum_{v=\alpha+1}^{n}(\bm{j_v})_{0...r_v-1}M_vN_v=\sum_{v=1}^{n}(\bm{j’_v})_{0...r_v-1}M_vN_v.
%    \end{aligned}
%\end{equation}
\begin{equation}\label{dindexij}
    \begin{aligned}
        i'=&m^*_L+\sum_{u=\alpha+1}^{n}(\bm{i_u^*})_{0...r_u-1}M_uN_u &:=\sum_{u=1}^{n}(\bm{i'_u})_{0...r_u-1}M_uN_u\\
        j'=&\sum_{v=1}^{\alpha}(\bm{j_v^*})_{0...r_v-1}M_vN_v+s^*_L &:=\sum_{v=1}^{n}(\bm{j’_v})_{0...r_v-1}M_vN_v.
    \end{aligned}
\end{equation}
For  $x_t$ and $y_t$, the corresponding $m_t^*,s_t^*$ in Eq.(\ref{sumofoc}) is
\begin{equation*}
    \begin{split}
        m_t^*=\sum_{u=1}^{n}(\bm{i'_u})_{0...r_{(u,y_t)}-1}M_uN_u,\quad
        s_t^*=\sum_{v=1}^{n}(\bm{j'_{v}})_{0...r_{(v,x_t)}-1}M_vN_v.
    \end{split}
\end{equation*}
Since $\rho\in \mathrm{KD}_{\mathcal{A,B}}^+$, by Lemma \ref{propofKD}, we get
\begin{equation*}
   \bra{a_{i'}}\rho\ket{b_{j'}}\braket{b_{j'}|a_{i'}}
        =|\braket{b_{j'}|a_{i'}}|^2\left(\lambda_{m_L^*,s_L^*}+\lambda_{m_{L-1}^*,s_{L-1}^*}+...+\lambda_{m_0^*,s_0^*}\right)\geq 0,\\ 
\end{equation*}
that means  

\begin{equation*}
       \lambda_{m_L^*,s_L^*}+\lambda_{m_{L-1}^*,s_{L-1}^*}+...+\lambda_{m_0^*,s_0^*}\geq 0.
\end{equation*}

We can get that all the coefficients of $\psi_{m_L,s_L}$ are nonnegative since  the arbitrariness of $m_L^*,s_L^*$. Since $\mathrm{Tr}(\rho)=1$, it follows directly that the sum of the coefficients in the new expression of $\rho$ is equal to one. Consequently, $\rho\in \mathrm{conv}(v_{path})$.
\end{proof}

%%=============================================%%
%% For submissions to Nature Portfolio Journals %%
%% please use the heading ``Extended Data''.   %%
%%=============================================%%

%%=============================================================%%
%% Sample for another appendix section			       %%
%%=============================================================%%

%% \section{Example of another appendix section}\label{secA2}%
%% Appendices may be used for helpful, supporting or essential material that would otherwise 
%% clutter, break up or be distracting to the text. Appendices can consist of sections, figures, 
%% tables and equations etc.

\end{appendices}

%%===========================================================================================%%
%% If you are submitting to one of the Nature Portfolio journals, using the eJP submission   %%
%% system, please include the references within the manuscript file itself. You may do this  %%
%% by copying the reference list from your .bbl file, paste it into the main manuscript .tex %%
%% file, and delete the associated \verb+\bibliography+ commands.                            %%
%%===========================================================================================%%

\bibliography{ref}% common bib file
%% if required, the content of .bbl file can be included here once bbl is generated
%%\input sn-article.bbl

\end{document}